\documentclass{article}
\usepackage{color}
\usepackage{amsmath, amsfonts, amsthm}
\usepackage{subfigure}
\usepackage{stmaryrd}
\usepackage{hyperref}
\usepackage[dvips]{graphicx}
\font\twelvemsb=msbm10 at 12pt
\newfam\msbfam
\textfont\msbfam=\twelvemsb

\newtheorem{theorem}{Theorem}[section]

\newtheorem{lemma}[theorem]{Lemma}
\newtheorem{proposition}[theorem]{Proposition}

\theoremstyle{definition}
\newtheorem{definition}[theorem]{Definition}
\newtheorem{remark}[theorem]{Remark}

\newcommand{\be}{\begin{equation}}
\newcommand{\ee}{\end{equation}}
\newcommand{\bq}{\begin{eqnarray}}
\newcommand{\eq}{\end{eqnarray}}

\newcommand{\ind}{1\hspace{-2.1mm}{1}} 
\newcommand{\I}{\mathtt{i}}
\newcommand{\D}{\mathrm{d}}
\newcommand{\E}{\mathrm{e}}
\newcommand{\sgn}{\mathrm{sgn}}
\usepackage{amssymb}

\begin{document}

\title{Asymptotic formulae for implied volatility in the Heston model
\thanks{
The authors would like to thank J.~Appleby, J.~Feng, J.P.~Fouque, J.~Gatheral, A.~Gulisashvili, M.~Keller-Ressel, A.~Lewis, R.~Lee and C.~Martini 
for many useful discussions and the anonymous referees for useful comments.
}}
\date{}
\author{ Martin Forde\thanks{Department of Mathematical Sciences, Dublin City University, {\tt martin.forde@dcu.ie}.}
         \and Antoine Jacquier\thanks{Department of Mathematics, Imperial College London and Zeliade Systems, Paris, {\tt ajacquie@imperial.ac.uk}.}
          \and Aleksandar Mijatovi\'{c}\thanks{Department of Mathematics, Imperial College London, {\tt a.mijatovic@imperial.ac.uk}.}}
\maketitle
\begin{abstract}
In this paper we prove an approximate formula expressed in terms of elementary functions for the implied volatility in the Heston model. 
The formula consists of the constant and first order terms in the large maturity expansion of the implied volatility function. 
The proof is based on saddlepoint methods and classical properties of holomorphic functions.
\end{abstract}

\section{Introduction}
In financial markets stochastic models are widely used by traders and risk managers to price and 
hedge financial products.
The models are chosen on economic grounds to reflect the observed characteristics of market data, 
such as leptokurtic returns of asset prices and random instantaneous volatility of these returns.
Pricing and hedging in realistic models
are numerically intensive procedures that need to be performed very quickly 
because future investment decisions depend on the
outcomes of these computations.
Therefore market practitioners 
pay particular attention to 
tractability when 
deciding which model to use.
Among the plethora of possible choices, stochastic volatility models are 
extremely popular particularly in equity, 
foreign exchange and interest rate markets since (i) they feature most of the  
characteristics of these markets, and (ii) they are numerically tractable.

In practice, stochastic volatility models are first calibrated on market data, then used for pricing. 
Pricing financial products is mathematically tantamount either to solving a PDE problem with boundary conditions 
(the final payoff of the product) or to calculating the expectation of this payoff using probabilistic tools such as Monte Carlo simulation or stochastic integration. 
For most models closed-form formulae are scarcely available, and accurate algorithms have been extended and used such as
finite-differences~\cite{Kluge}, ADI schemes~\cite{Foulon}, or quadrature~\cite{AMST06} methods. 
The calibration step involves a proper selection of the data to be fitted by a model. 
A common practice is to calibrate the so-called implied volatility rather than option prices directly.
The implied volatility is a standardised measure of option prices which makes them comparable even though the underlying assets are not the same.
Since this calibration step is based on optimisation algorithms, the lack of a closed-form formula for the implied volatility makes it very time consuming. 
For instance, the SABR stochastic volatility model has become very popular because a closed-form approximation formula for the implied volatility was derived in \cite{SABR} and hence made the model easily tractable. 
Likewise, perturbation methods as developed in \cite{FPS00} have proved to be very useful for obtaining a closed-form approximation formula of option prices. 
Although these methods only hold under some constraints on the parameters, they provide useful initial reference points for calibration.

The Heston model~\cite{Heston} introduced in 1993 has become one of the most widely used stochastic volatility models in the derivatives market (see~\cite{Gatheral}, \cite{Lewis00}, \cite{AMST06}, \cite{AP07}, \cite{LK08}). 
In this paper, we provide a closed-form approximation for the implied volatility in this model.
The idea behind this result is the following: 
suppose one wants to calibrate the Heston model on market data.
This can be performed in two different ways: (i) one uses a global optimisation algorithm, 
which involves computing the implied volatility at each observed point until the algorithm converges; 
(ii) one specifies an initial set of parameters for the model and runs a local optimisation algorithm such as the least-squares method. 
The latter solution is the most widely used in practice since it is less computer-intensive.
However its robustness heavily relies on the initial set of parameters to be specified.
Simple closed-form approximations for the Heston model make this choice robust and accurate.
One first calibrates the approximation on market data---which is straightforward since this is a closed-form---then one uses this calibrated set of parameters as a starting point in the whole calibration process.

Let us consider an European option with maturity $t$ and maturity-dependent strike $K=S_0 \exp(xt)$, 
then our main result is the following asymptotic closed-form formula for the implied volatility 
$\hat{\sigma}_t^2(x)$:
\begin{equation}\label{eq:Intro1}
\hat{\sigma}_t^2(x)=\hat{\sigma}_{\infty}^2(x)+t^{-1}\frac{8\hat{\sigma}_{\infty}^4(x)}{4x^2-\hat{\sigma}_{\infty}^4(x)}\log\left(\frac{A(x)}{A_{\mathrm{BS}}(x,\hat{\sigma}_{\infty}(x),0)}\right)+o\left(t^{-1}\right)
\end{equation}
as the maturity $t$ tends to infinity, where $\hat{\sigma}_\infty^2$ is defined in \eqref{DefOfSigmaInf}, $A$ in \eqref{eq:AHeston} and $A_{\mathrm{BS}}$ in \eqref{eq:DefOfABS}. For a constant strike $K=S_0\exp(x)$, we obtain the following formula:
\begin{equation}\label{eq:Intro2}
\sigma_t^2(x)= 8V^*(0)+t^{-1}4\left(x\left(2\,p^*(0)-1\right)-2\log\left(-A(0)\sqrt{2V^*(0)}\right)\right)+o\left(t^{-1}\right)
\end{equation}
as the maturity $t$ tends to infinity, where $V^*$ is given by \eqref{DefOfVStar} and $p^*$ by \eqref{Saddlepoint}.

It is a well-known fact that for a fixed strike, the implied volatility flattens as the maturity increases~\cite{Rogers}; 
this is confirmed by formula \eqref{eq:Intro2} above, the zeroth order term of which was already known 
(see \cite{Lewis00}, \cite{FJ09II}). However, the maturity-dependent strike formulation in formula \eqref{eq:Intro1} above reveals that the implied volatility smile does not flatten but rather spreads out in a very specific way as the maturity increases.

In the fixed-strike case, Lewis~\cite{Lewis00} pioneered the research on large-time asymptotics 
of implied volatility in stochastic volatility models by studying 
the first eigenvalue and eigenfunction of the generator of the underlying 
stochastic process.
Recently Tehranchi \cite{Tehr09} studied the large-time behaviour of the implied volatility 
when the stock price is a non-negative local martingale and obtained an analogue of 
formula~\eqref{eq:Intro2} in that setting. 
Comparatively, there has been a profusion of work on small-time asymptotics, 
based on differential geometry techniques~\cite{Labordere}, 
PDE methods \cite{Berestycki} or large deviations techniques (\cite{FJ09I} and \cite{FFF09}). 
Likewise, many papers have studied the behaviour of the implied volatility smile in the 
wings (see \cite{BF1}, \cite{BF2}, \cite{Gulisashvili1}, \cite{Gulisashvili2}, \cite{Lee042}).

The proof of our main result, Theorem \ref{thm:HestonLargeT}, is based on two methods: 
first, we use saddlepoint approximation methods to study the behaviour of the call price 
function as an inverse Fourier transform. This idea has already been applied by several authors, including \cite{CarrMadan}, \cite{Glasserman}, \cite{AitSahalia} and \cite{RogersSaddle} in order to speed up the computation of option pricing algorithms based on inverse Fourier transforms. We are also able to obtain the saddlepoint in closed form, thus avoiding any numerical approximations in determining it. The second step in our proof relies on Cauchy's integral theorem and contour integration for holomorphic functions in order to obtain precise estimates of call option prices in the large maturity limit.

The paper is organised as follows. Section \ref{section:notations} contains the large-time asymptotic formula for call options under the Heston and the Black-Scholes
models, both in the maturity-dependent and in the fixed-strike case. 
The proof of the main theorem, Theorem \ref{thm:HestonLargeT}, is given in Section \ref{ProofMainTheorem}. In section \ref{section:ResultsImpliedVol}, we translate these results into implied volatility asymptotics and prove formulae \eqref{eq:Intro1} and \eqref{eq:Intro2} above. In Section \ref{section:Numerics}, we calibrate the Heston model and provide numerical examples based on formulae \eqref{eq:Intro1} and \eqref{eq:Intro2}.

\section{Large-time behaviour of call options}\label{section:notations}

Throughout this article, we work on a model $(\Omega,\mathcal{F},P)$ with a filtration $(\mathcal{F}_t)_{t \ge 0}$ supporting two Brownian motions, and satisfying the usual conditions. Let $(S_t)_{t\geq 0}$ denote a stock price process and we let $X_t:=\log(S_t)$. Interest rates and dividends are considered null. We assume the following Heston dynamics for the log-stock price:
\begin{equation}\label{eq:HestonModel}
        \begin{array}{ll}
        \D X_{t}=-\frac{1}{2}Y_t\D t+\sqrt{Y_{t}}\D W^1_{t},\,\, X_0=x_0\in\mathbb{R},\\
        \D Y_{t}=\kappa\left(\theta-Y_t\right)\D t+\sigma\sqrt{Y_t}\D W^2_{t},\,\, Y_0=y_0>0,\\
        \D\langle W^1, W^2\rangle_{t}=\rho \D t
        \end{array}\
\end{equation}
with $\kappa,\theta,\sigma,y_0>0,|\rho|<1$.

The Feller condition $2\kappa \theta>\sigma^2$ ensures that $0$ is an unattainable boundary for the process $Y$. 
If this condition is violated zero is an attainable, regular and reflecting
boundary (see chapter 15 in~\cite{KarlinTaylor} for the classification of boundary
points of one-dimensional diffusions). 
Since the analysis in this paper relies solely on the study of the behaviour of the 
Laplace transform of the process $X$, 
which remains well defined even if the Feller condition is not satisfied,
we do not assume  that the inequality
$2\kappa \theta>\sigma^2$
holds.

Let us now define $\bar{\kappa}:=\kappa-\rho\sigma$, $\bar{\rho}:=\sqrt{1-\rho^2}$\label{defofrhobar}, and $\bar{\theta}:=\kappa\theta/\bar{\kappa}$. 
Throughout the whole paper, we will assume $\bar{\kappa}>0$. 
This assumption ensures (see Theorem 2.1 in~\cite{FJ09II}) that moments of $S$ greater than $1$ exist for all times $t$. 
This condition is fundamental for the analysis in the paper and is 
usually assumed in the literature 
(see \cite{KR08} and \cite{AP07}). 
When this condition is violated the limiting logarithmic Laplace transform
$V$ defined in~\eqref{DefOfFunctionV} of the process $S$ does not have the
same properties, 
and further research is needed to understand how the implied volatility behaves in this case. 
We know from~\cite{AP07} and~\cite{FJ09II} that $\bar{\kappa}$ is the mean-reversion level of the 
process $Y$ 
under the so-called Share measure that is equivalent to the original
probability measure with the Radon-Nikodym  derivative given by the share price 
itself.  
If $\bar{\kappa}\leq 0$, the  process $Y$ 
will be neither mean-reverting nor ergodic under the Share measure. 
In the equities market this does not constitute a problem since the calibrated 
correlation $\rho$ is always negative. However this assumption may be
restrictive in markets such as foreign exchange, and further research is required
to relax it.
Let $V$ be the limiting logarithmic moment generating function of $X$ defined as
\begin{equation}\label{DefOfFunctionV}
V(p):=\lim_{t\to\infty}t^{-1}\log\mathbb{E}\left(\exp\Big(p\left(X_t-x_0\right)\Big)\right),
\end{equation}
for all $p$ such that the limit exists and is finite. It follows from \cite{AP07} that $V$ is a well defined and strictly convex function on 
$\left(p_-,p_+\right)$ 
and is 
infinite outside, where
\begin{equation}\label{DefOfPpm}
p_{\pm}:=\left(-2 \kappa \rho+\sigma \pm \sqrt{\sigma^2+4 \kappa^2-4 \kappa \rho\sigma}\right)/\left(2\sigma\bar{\rho}^2\right),
\end{equation}
with $p_-<0$ and $p_+>1$. Furthermore the function $V$ takes the following form
\begin{equation}\label{eq:V(p)}
V(p)=\frac{\kappa\theta}{\sigma^2}\Big(\kappa-\rho\sigma
p-d(-\I p)\Big),\quad\text{for }p\in\left(p_-,p_+\right),
\end{equation}
where
\begin{equation}\label{DefinitionOfD}
d(k):=\sqrt{(\kappa-\I \rho \sigma k)^2 + \sigma^2(\I k+ k^2)},\quad\text{for }k\in\mathbb{C},
\end{equation}
and we take the principal branch for the complex square root function in \eqref{DefinitionOfD}.

Let us now define the Fenchel-Legendre transform 
$V^*(x):=\sup\{px-V(p):\>p\in(p_-,p_+)\}$
of
$V$,
which was computed in~\cite{FJ09II}
and is given by the formula
\begin{equation}\label{DefOfVStar}
V^*(x)=p^*(x)x-V(p^*(x)),\quad\text{for all }x\in\mathbb{R},
\end{equation}
where the function $p^*:\mathbb{R}\to\left(p_-,p_+\right)$ is defined by 
\begin{equation}\label{Saddlepoint}
p^*(x):=\frac{\sigma-2\kappa\rho+(\kappa\theta\rho+x\sigma)\left((\sigma^2+4\kappa^2-4\kappa\rho\sigma)/(x^2\sigma^2+2x\kappa\theta\rho\sigma+\kappa^2\theta^2)\right)^{1/2}}{2\sigma\bar{\rho}^2},\quad\text{for all }x\in\mathbb{R}.
\end{equation}
Tedious but straightforward calculations using the explicit formulae above
yield the following proposition.
\begin{proposition}\label{PropIndicators}
The function $p^*:\mathbb{R}\to(p_-,p_+)$, where $p_-$ and $p_+$ are defined in \eqref{DefOfPpm}, is strictly increasing, infinitely differentiable and satisfies the following properties
$$p^*\left(-\theta/2\right)=0,\quad p^*\left(\bar{\theta}/2\right)=1,\quad \lim\limits_{x\to-\infty}p^*(x)=p_-\quad \text{and}\lim\limits_{x\to+\infty}p^*(x)=p_+,$$
as well as the equation
\begin{equation}\label{VPrimeOfpStar}
V'\left(p^*(x)\right)=x,\quad\text{for all }x\in\mathbb{R}.
\end{equation}
\end{proposition}

Since the image of $p^*$ is  $\left(p_-,p_+\right)$, the function $V^*$ is well defined on $\mathbb{R}$.
The following properties of $V^*$ are easy to prove \label{PropertiesOfVStar}
and will be used throughout the paper:
\begin{itemize}
\item[(a)] $V^{*'}(x)=p^*(x)$ for all $x\in\mathbb{R}$;
\item[(b)] $V^{*''}(x)>0$ for all $x\in\mathbb{R}$;
\item[(c)] $x\mapsto V^*(x)$ is non-negative, has a unique minimum at $-\theta/2$ and $V^*(-\theta/2)=0$;
\item[(d)] $x\mapsto V^*(x)-x$ is non-negative, has a unique minimum at $\bar{\theta}/2$ and $V^*\left(\bar{\theta}/2\right)=\bar{\theta}/2$.
\end{itemize}
From the definition \eqref{DefOfVStar} of $V^*$ and relation \eqref{VPrimeOfpStar}, 
the equality in (a) follows. The inequality in (b) is a consequence of (a) and Proposition \ref{PropIndicators}. Now, (a), (b) and Proposition \ref{PropIndicators} imply that $-\theta/2$ is the only local minimum of the function $V^*$ and is therefore a global minimum. The definition of $V^*$ given in \eqref{DefOfVStar} implies $V^*\left(-\theta/2\right)=-V(0)=0$. Since the stock price $S$ is a true martingale (see~\cite{AP07}), we have $V(1)=0$ and Proposition \ref{PropIndicators} implies that $V^*\left(\bar{\theta}/2\right)=\bar{\theta}/2>0$. This proves (c). From (a) and Proposition \ref{PropIndicators}, we know that the function $x\mapsto V^{*}(x)-x$ has a unique minimum attained at $\bar{\theta}/2$ and $V^*\left(\bar{\theta}/2\right)-\bar{\theta}/2=0$. Therefore (b) implies (d). 

\subsection{Large-time behaviour of call options under the Heston model}
In this section, we derive the asymptotic behaviour of call option prices under the Heston dynamics \eqref{eq:HestonModel} as the maturity $t$ tends to infinity, 
both for maturity-dependent and for fixed strikes. Before diving into the core of this paper, let us introduce the function 
$\mathcal{I}:\mathbb{R}\times\mathbb{R}_+\times\mathbb{R}^2\to\mathbb{R}$ by
\begin{equation}\label{eq:DefOfI}
\mathcal{I}\left(x,t;a,b\right):=\left(1-\E^{xt}\right)\ind_{\left\{x<a\right\}}+\ind_{\left\{a<x<b\right\}}+\frac{1}{2}\ind_{\left\{x=b\right\}}+\left(1-\E^{at}\right)\ind_{\left\{x=a\right\}},
\end{equation}
which will feature in the main formulae of Theorem~\ref{thm:HestonLargeT}
and Proposition~\ref{prop:BStimedep}.
The next theorem is the main result of the paper and its proof is given in Section \ref{ProofMainTheorem}.
\begin{theorem}\label{thm:HestonLargeT}
For the Heston model and the assumptions above,  we have the following asymptotic behaviour for the price of a call option with strike $S_0\exp(xt)$ for all $x\in\mathbb{R}$,
$$\frac{1}{S_0}\mathbb{E}\left(S_t-S_0\E^{xt}\right)^{+}=\mathcal{I}\left(x,t;-\theta/2,\bar{\theta}/2\right)
+(2\pi t)^{-1/2}\exp\Big(-\left(V^*(x)-x\right)t\Big)A(x)\left(1+O\left(1/t\right)\right),\quad\text{as }t\to\infty,$$
where
\begin{equation}\label{eq:AHeston}
A(x) := \frac{1}{\sqrt{V''\left(p^*(x)\right)}}\left\{
            \begin{array}{ll}
\displaystyle \frac{U(p^*(x))}{p^{*}(x)\left(p^{*}(x)-1\right)},&\text{if }x\in\mathbb{R}\setminus\left\{-\theta/2,\bar{\theta}/2\right\},\\
\\
\displaystyle -1-\sgn\left(x\right)\left(\frac{1}{6}\frac{V'''\left(p^*(x)\right)}{V''\left(p^*(x)\right)}-U'\left(p^*(x)\right)\right),&\text{if }x\in\left\{-\theta/2,\bar{\theta}/2\right\},\\
            \end{array}
            \right.
\end{equation}
where
\begin{equation}\label{eq:DefOfU}
U(p):= \left(\frac{2d\left(-\I p\right)}{\kappa-\rho\sigma p+d\left(-\I p\right)}\right)^{2\kappa\theta/\sigma^2}\exp\left(\frac{ y_0}{\kappa \theta}V(p)\right),
\end{equation}
$V$ is defined in \eqref{eq:V(p)}, $p^*$ in \eqref{Saddlepoint}, $V^*$ in \eqref{DefOfVStar} and $d$ in \eqref{DefinitionOfD}
and 
$\sgn(x)$
equals
$1$ 
if 
$x$
is positive and 
$-1$
otherwise.
\end{theorem}
\begin{remark}
\label{rem:Poles_Of_A}
Property (b) on page~\pageref{PropertiesOfVStar} implies that the square root 
$\sqrt{V''\left(p^*(x)\right)}$
is a strictly positive real number. 
Note also that the function 
$A$ 
defined in~\eqref{eq:AHeston}
is not continuous at the points 
$-\theta/2$ and $\bar{\theta}/2$ 
by Proposition~\ref{PropIndicators}
(see also Figure \ref{FigureAandABS} on page \pageref{FigureAandABS}).
\end{remark}
\begin{remark}
It is proved in Proposition \ref{PropIndicators} that $p^*\left(-\theta/2\right)=0$ and $p^*\left(\bar{\theta}/2\right)=1$. Note further that $U'(0)=(\theta-y_0)/(2\kappa)$ and $U'(1)=(y_0-\bar{\theta})/(2\bar{\kappa})$.
\end{remark}
\begin{remark} Theorem \ref{thm:HestonLargeT} is similar in spirit to the saddlepoint approximation for a density of random variable $X$ given in Butler \cite{But}
$$f_X(x) \approx \Big(2\pi K''(p^*(x))\Big)^{-1/2}\exp\Big(K(p^*(x))-xp^*(x)\Big),$$
where $K(p):=\log\mathbb{E}(\exp(pX))$, and $p^*(x)$ is the unique solution to $K'(p^*(x))=x$. Here $X:=X_t-x_0$, $K(p)=tV(p)+O(1)$, and we substitute $x$ to $xt$ so that
$$f_{X_t-x_0}(xt) \approx \left(2\pi V''\left(p^*(x)\right)t\right)^{-1/2}\exp\left(V\left(p^*(x)\right)t-xp^*(x)t\right)= \left(2\pi V''\left(p^*(x)\right)t\right)^{-1/2}\exp\left(-V^*\left(x\right)t\right).$$
\end{remark}
\noindent In order to precisely compare our result to the existing literature, we prove the following lemma, which gives the asymptotic behaviour of vanilla call options when the strike $K$ is fixed, independent of the maturity. 
The following lemma was derived in \cite{Lewis00}, Chapter 6; a rigorous proof is detailed here in Appendix~\ref{section:ProofofFixedStrike}.
\begin{lemma}
\label{lem:LargeTFixedK}
Under the same assumptions as Theorem \ref{thm:HestonLargeT}, for any $x\in\mathbb{R}$, we have the following behaviour for a call option with fixed strike $K=S_0\,\exp(x)$,
$$\frac{1}{S_0}\mathbb{E}\left(S_t-K\right)^{+}=1+\frac{A(0)}{\sqrt{2\pi t}}\exp\Big((1-p^*(0))x-V^*(0)t\Big)\left(1+O(1/t)\right),\text{ as }t\to\infty.$$
\end{lemma}
\subsection{Large-time behaviour of the the Black-Scholes call option formula}
By a similar analysis, we can deduce the large-time asymptotic call price for the Black-Scholes model. 
This result is of fundamental importance for us since it will allow us to compute the implied volatility 
by comparing the Black-Scholes and the Heston call option prices. Throughout the rest of the paper, we let $C_{\mathrm{BS}}(S_0,K,t,\Sigma)$ denote the Black-Scholes price of a European call option written on a reference stock price $S$, with strike $K>0$, initial stock price $S_0>0$, time-to-maturity $t\geq 0$, and volatility $\Sigma>0$ (with zero interest rate and zero dividend). 
Similar to Section \ref{section:notations}, let us define the function $V_{\mathrm{BS}}:\mathbb{R}\to\mathbb{R}$ as in \eqref{DefOfFunctionV}, where $X_t:=\log(S_t)$. In the Black-Scholes case, it reads
\begin{equation}\label{eq:DefOfVbs}
V_{\mathrm{BS}}(p)=p\left(p-1\right)\Sigma^2/2,\quad\text{for all }p\in\mathbb{R}.
\end{equation}
Similarly to \eqref{DefOfVStar}, we can define the functions $V^*_{\mathrm{BS}}:\mathbb{R}\times\mathbb{R}_+^*\to\mathbb{R}$ and $p_{\mathrm{BS}}^*:\mathbb{R}\to\mathbb{R}$, by
\begin{equation}\label{eq:DefOfVStarBS}
V^*_{\mathrm{BS}}\left(x,\Sigma\right):=\left(x+\Sigma^2/2\right)^2/\left(2\Sigma^2\right),\quad\text{for all }x\in\mathbb{R},\ \Sigma\in\mathbb{R}_+^*,
\end{equation}
and
\begin{equation}\label{eq:DefOfpStarBS}
p^*_{\mathrm{BS}}(x):=\left(x+\Sigma^2/2\right)/\Sigma^2,\quad\text{for all }x\in\mathbb{R}.
\end{equation}
The following proposition, proved in Appendix \ref{ProofLemma24}, gives the behaviour of the Black-Scholes price as the maturity tends to infinity.
\begin{proposition}\label{prop:BStimedep}
Let $\sigma>0$ and 
let the real number 
$a_1$
satisfy
$a_1>-\sigma^2 t$ 
for large times 
$t$.
Then for all $x\in\mathbb{R}$, 
we have the following asymptotic behaviour for the Black-Scholes call option formula in the large-strike, large-time case
$$
\frac{1}{S_0}C_{\mathrm{BS}}\left(S_0,S_0 \E^{xt},t,\sqrt{\sigma^2+\frac{a_1}{t}}\right)
= \mathcal{I}\left(x,t;-\frac{\sigma^2}{2},\frac{\sigma^2}{2}\right)
+\frac{A_{\mathrm{BS}}(x,\sigma,a_1)}{\sqrt{2\pi t}}\E^{-\left(V_{\mathrm{BS}}^*(x,\sigma)-x\right)t}(1+ O(1/t)),
$$
where
\begin{equation}\label{eq:DefOfABS}
A_{\mathrm{BS}}(x,\sigma,a_1)  := \exp\left(\frac{1}{8}a_1\left(\frac{4x^2}{\sigma^4}-1\right)\right)\frac{\sigma^3}{x^2-\sigma^4/4}\ind_{\{x\ne\pm\sigma^2/2\}}+\frac{a_1/2-1}{\sigma }\ind_{\{x=\pm\sigma^2/2\}},
\end{equation}
and where the function $\mathcal{I}$ is defined in \eqref{eq:DefOfI}.
\end{proposition}
\begin{remark}
\label{rem:BSLargeT}
If we set 
$a_1=0$ 
in Proposition \ref{prop:BStimedep}, we obtain the large-time expansion for a call option under the standard Black-Scholes model with volatility $\sigma$ and log-moneyness equal to $xt$.
\end{remark}
\noindent As in the Heston model above, we derive here the equivalent of Proposition \ref{prop:BStimedep} when the strike does not depend on the maturity anymore.
\begin{lemma}\label{lem:FixedStrikeOptionBSLemma}
With the assumptions above, we have the following behaviour for the Black-Scholes call option formula in the fixed-strike, large-time case
$$\frac{1}{S_0}C_{\mathrm{BS}}\left(S_0,S_0 \E^{x},t,\sqrt{\sigma^2+a_1/t}\right)=1-\frac{2\sqrt{2}}{\sigma\sqrt{\pi  t}}\exp\left(-\sigma^2 t/8+x/2-a_1/8\right)\left(1+O(1/t)\right).$$
\end{lemma}
This lemma is immediate from the Black-Scholes formula given in Appendix \ref{ProofLemma24} and the approximation \eqref{OlverApproxPhi} for the Gaussian cumulative distribution function.
\section{Large-time behaviour of implied volatility}\label{section:ResultsImpliedVol}
The previous section dealt with large-time asymptotics for call option
prices.  In this section, we translate these results into asymptotics for
the implied volatility. Recall that \cite{FJ09II} and \cite{Lewis00} have already derived the leading order term for the implied
volatility in the large-time, fixed-strike case. Our goal here
is to obtain the leading order and the correction term in the large-time,
large-strike case. Theorem \ref{thm:ImpliedVol} provides the main result, i.e. the large-time behaviour of the implied volatility in the large strike case. In the following, $\hat{\sigma}_t(x)$ will denote the implied volatility corresponding to a vanilla call option with maturity $t$ and (maturity-dependent) strike $S_0\exp(xt)$ in the Heston model \eqref{eq:HestonModel}. We now define the functions $\hat{\sigma}_\infty^2:\mathbb{R}\to\mathbb{R}_+$ and $\hat{a}_1:\mathbb{R}\to\mathbb{R}$ by
\begin{equation}\label{DefOfSigmaInf}
\hat{\sigma}^2_{\infty}(x):=2\left(2V^*(x)-x+2
\left(\ind_{x\in(-\theta/2,\bar{\theta}/2)}-\ind_{x\in\mathbb{R}\setminus(-\theta/2,\bar{\theta}/2)}\right)\sqrt{V^*(x)^2-V^*(x)x}\right),
\text{ for all }x\in\mathbb{R},
\end{equation}
and
\begin{equation}\label{eq:Eqa1}
\hat{a}_1(x):=2\left\{
        \begin{array}{ll}
\displaystyle          \Big(x^2/\hat{\sigma}_{\infty}^4(x)-1/4\Big)^{-1}\log\Big(A(x)/A_{\mathrm{BS}}(x,\hat{\sigma}_{\infty}(x),0)\Big),\quad\quad\quad\quad\quad x\in\mathbb{R}\setminus\left\{-\frac{\theta}{2},\frac{\bar{\theta}}{2}\right\},\\ \\
\displaystyle 1-\frac{\hat{\sigma}_\infty(x)}{\sqrt{V''(p^*(x))}}\left(1+\sgn(x)\left(\frac{V'''(p^*(x))}{6V''(p^*(x))}-U'(p^*(x))\right)\right),\quad x\in\left\{-\frac{\theta}{2},\frac{\bar{\theta}}{2}\right\},
        \end{array}\
        \right.
\end{equation}
where $A$ is defined in \eqref{eq:AHeston}, $A_{\mathrm{BS}}$ in \eqref{eq:DefOfABS}, $U$ in \eqref{eq:DefOfU}, $V$ in \eqref{eq:V(p)},  $p^*$ in \eqref{Saddlepoint} and $V^*$ in \eqref{DefOfVStar}. They are all completely explicit, so that the functions $\hat{\sigma}^2_{\infty}$ and $\hat{a}_1$ are also explicit. From the properties of $V^*$ proved on page \pageref{PropertiesOfVStar}, $V^*(x)$ and $V^*(x)-x$ are non-negative, so that $\hat{\sigma}^2_\infty(x)$ is a well defined real number for all $x\in\mathbb{R}$. Then the following theorem holds.
\begin{theorem}
\label{thm:ImpliedVol}
The functions $\hat{\sigma}_\infty$ and $\hat{a}_1$ are continuous on $\mathbb{R}$ and
$$\hat{\sigma}_t^2(x)=\hat{\sigma}_{\infty}^2(x)+\hat{a}_1(x)/t+o\left(1/t\right),\quad\text{for all }x\in\mathbb{R},\text{ as }t\to\infty.$$
Furthermore the error term
$|\hat{\sigma}^2_t(x)-\hat{\sigma}^2_\infty(x)-\hat{a}_1(x)/t|t$
tends to zero as 
$t$
goes to infinity 
uniformly in
$x$
on compact subsets of 
$\mathbb{R}\backslash\{-\theta/2,\bar{\theta}/2\}$.
\end{theorem}
\begin{proof}
We first prove that the functions 
$\hat{\sigma}_\infty$ 
and 
$\hat{a}_1$ are continuous. In fact, the continuity of the function $\hat{\sigma}_\infty$ follows from properties (c) and (d) on 
page~\pageref{PropertiesOfVStar}. 
We have already observed that the function $A$ defined 
in~\eqref{eq:AHeston} 
is discontinuous 
at $\bar{\theta}/2$ and $-\theta/2$ 
(see Remark~\ref{rem:Poles_Of_A}).
Elementary calculations show that in the neighbourhood of 
$\bar{\theta}/2$ we have the following expansion 
(where
$\Theta:=(\bar{\theta}/V''(1))^{1/2}$)
$$\hat{\sigma}^2_\infty(x)=\bar{\theta}+2\left(1-\Theta\right)\left(x-\bar{\theta}/2\right)+
\frac{2}{V''(1)}\left(1-\frac{1}{\Theta}+\frac{V'''(1)}{6V''(1)^{2}}\Theta\right)
\left(x-\bar{\theta}/2\right)^2
+
O\left(\left|x-\bar{\theta}/2\right|^{3}\right).$$
This expansion, together with formula~\eqref{eq:DefOfABS},
implies that the function
$x\mapsto A_{\mathrm{BS}}\left(x,\hat{\sigma}_\infty(x),0\right)$ 
is also discontinuous at 
$\bar{\theta}/2$ and $-\theta/2$. 
Since 
the quotient 
$x\mapsto A(x)/A_{\mathrm{BS}}(x,\hat{\sigma}_\infty(x),0)$ 
is strictly positive on $\mathbb{R}\setminus\left\{-\theta/2,\bar{\theta}/2\right\}$, 
the function $\hat{a}_1$ is therefore continuous on this set. 
The expansion above can also be used to obtain the following expression
$$
\frac{A(x)}{A_{\mathrm{BS}}(x,\hat{\sigma}_\infty(x),0)} = 1 +\frac{1}{V''(1)}
\left(U'(1)-1-\frac{V'''(1)}{6V''(1)} +\frac{1}{\Theta} \right)(x-\bar{\theta}/2)+O\left(|x-\bar{\theta}/2|^2\right),
$$
which implies 
the equality
$\lim\limits_{x\to\bar{\theta}/2}\hat{a}_1(x)=\hat{a}_1\left(\bar{\theta}/2\right)$.
A similar argument shows continuity of $\hat{a}_1$ at $-\theta/2$.

We now prove the formula in the theorem in the case $x>\bar{\theta}/2$
(recall that $\bar{\theta}>0$). 
Note that $\hat{\sigma}_{\infty}(x)$ as defined in \eqref{DefOfSigmaInf}  satisfies the following quadratic equation
\begin{equation}\label{eq:Quadratic}
V^*(x)-x=V_{\mathrm{BS}}^*(x,\hat{\sigma}_{\infty}(x))-x,\quad\text{for all }x\in\mathbb{R},
\end{equation}
where $V^*$ is given by \eqref{DefOfVStar} and $V^*_{\mathrm{BS}}$ by \eqref{eq:DefOfVStarBS}. 
The proof of the theorem is in two steps: first, we prove the uniform convergence (on small intervals) 
of the implied variance to 
$\hat{\sigma}^2_{\infty}(x)$,
then we derive a similar result for the first order correction term. 
As a first step, we have to prove that, for all $\delta>0$, 
there exists $t^*(\delta)>0$ 
such that for all 
$t>t^*(\delta)$
and all
$x$
in some small neighbourhood of 
$x'>\bar{\theta}/2$
we have
$|\hat{\sigma}_t(x)-\hat{\sigma}_{\infty}(x)|<\delta$.
By Theorem \ref{thm:HestonLargeT} and \eqref{eq:Quadratic} 
we know that for all $\epsilon>0$, 
there exists $t^*(\epsilon)$ 
and an interval containing 
$x'$
such that for all 
$x$
in this interval 
and all
$t>t^*(\epsilon)$ 
we have the lower bound
\begin{equation}\label{eq:lowerbound}
\exp\Big(-\left(V_{\mathrm{BS}}^*(x,\hat{\sigma}_{\infty}(x))-x+\epsilon\right)t\Big)= \exp\Big(-\left(V^*(x)-x+\epsilon\right)t\Big) 
<\frac{1}{S_0}\mathbb{E}\left(S_t-S_0\E^{xt}\right)^{+},
\end{equation}
and the upper bound
\begin{equation}\label{eq:upperbound}
\frac{1}{S_0}\mathbb{E}\left(S_t-S_0\E^{xt}\right)^{+}< \exp\Big(-\left(V^*(x)-x-\epsilon\right)t\Big) =\exp\Big(-\left(V_{\mathrm{BS}}^*(x,\hat{\sigma}_{\infty}(x))-x-\epsilon\right)t\Big).
\end{equation}
Note that
$$\hat{\sigma}^2_{\infty}(x)-2x=4\Big((V^*(x)-x)-\sqrt{(V^*(x)-x)^2+(V^*(x)-x)x}\Big)<0,\quad\text{for all}\quad x\in\mathbb{R}\backslash(-\theta/2,\bar{\theta}/2),$$
since $V^*(x)-x> 0$ by property (d) on page \pageref{PropertiesOfVStar}. For $x$ fixed, the function $\Sigma\mapsto V_{\mathrm{BS}}^*(x,\Sigma)-x$ defined on
$(0,\sqrt{2x})$ is continuous 
and strictly decreasing, where $V^*_{\mathrm{BS}}$ is given in \eqref{eq:DefOfVStarBS}. 
Thus, for any $\delta>0$ such that 
$\hat{\sigma}_\infty(x)\pm\delta\in(0,\sqrt{2x})$ 
for all 
$x$
in 
the small neighbourhood of 
$x'$
define
\begin{align*}
\epsilon_1(\delta) & := \Big(V_{\mathrm{BS}}^*(x,\hat{\sigma}_{\infty}(x))-V_{\mathrm{BS}}^*(x,\hat{\sigma}_{\infty}(x)+\delta)\Big)/2>0,\\
\epsilon_2(\delta) & := \Big(V_{\mathrm{BS}}^*(x,\hat{\sigma}_{\infty}(x)-\delta)-V_{\mathrm{BS}}^*(x,\hat{\sigma}_{\infty}(x))\Big)/2>0.
\end{align*}
Note that the limits
$\lim\limits_{\delta\to 0}\epsilon_1(\delta)=\lim\limits_{\delta\to 0}\epsilon_2(\delta)=0$ 
are uniform in 
$x$
on the chosen neighbourhood of  
$x'$.
Combining \eqref{eq:lowerbound}, \eqref{eq:upperbound} and Proposition \ref{prop:BStimedep}, there exists $t^*(\delta)$ such that for all $t>t^*(\delta)$
and all 
$x$
near
$x'$
we have
$$\frac{1}{S_0}C_{\mathrm{BS}}\left(S_0,S_0\E^{xt},t,\hat{\sigma}_\infty(x)-\delta\right)\le 
\exp\Big(-(V_{\mathrm{BS}}^*(x,\hat{\sigma}_{\infty}(x)-\delta)-x-\epsilon_2(\delta))t\Big)<
\frac{1}{S_0}\mathbb{E}\left(S_t-S_0\E^{xt}\right)^{+},$$
and
$$\frac{1}{S_0}\mathbb{E}\left(S_t-S_0\E^{xt}\right)^{+} <
\exp\Big(-(V_{\mathrm{BS}}^*(x,\hat{\sigma}_{\infty}(x)+\delta)-x+\epsilon_1(\delta))t\Big)\le
\frac{1}{S_0}C_{\mathrm{BS}}\left(S_0,S_0 \E^{xt},t,\hat{\sigma}_\infty(x)+\delta\right).$$
Thus, by the monotonicity of the Black-Scholes call option formula as a function of the volatility, 
we have the following bounds for the implied volatility $\hat{\sigma}_t(x)$ at maturity $t$
$$\hat{\sigma}_{\infty}(x)-\delta<\hat{\sigma}_t(x)<\hat{\sigma}_{\infty}(x)+\delta.$$

In the second step of the proof we show that for all $\delta>0$ there exist $t^*(\delta)$ and a small
neighbourhood of 
$x'$
such that, for all $t>t^*(\delta)$ and all $x$ in this neighbourhood,
the following holds
\begin{equation}\label{eq:ToProveCorrectionTerm}
|\hat{\sigma}^2_t(x)-\hat{\sigma}^2_\infty(x)-\hat{a}_1(x)/t|<\delta/t.
\end{equation}
Note that the definition in \eqref{eq:Eqa1} implies
$$A(x)=A_{\mathrm{BS}}(x,\hat{\sigma}_{\infty}(x),\hat{a}_1(x)),\quad\text{for all }x\in\mathbb{R}.$$
Theorem \ref{thm:HestonLargeT} implies that 
for any $\epsilon>0$ there exist 
a small interval containing 
$x'$
and 
$t^*(\epsilon)$ such that for all 
$x$
in this interval and all
$t>t^*(\epsilon)$ we have
\be
\label{eq:eq1}
\frac{1}{S_0}\mathbb{E}\left(S_t-S_0\E^{xt}\right)^{+} < \frac{A_{\mathrm{BS}}(x,\hat{\sigma}_{\infty}(x),\hat{a}_1(x))}{\sqrt{2\pi t}}\E^{-(V_{\mathrm{BS}}^*(x,\hat{\sigma}_{\infty}(x))-x)t}\E^{\epsilon},
\ee
and
$$\frac{A_{\mathrm{BS}}(x,\hat{\sigma}_{\infty}(x),\hat{a}_1(x))}{\sqrt{2\pi t}}\E^{-(V_{\mathrm{BS}}^*(x,\hat{\sigma}_{\infty}(x))-x)t}\E^{-\epsilon}< \frac{1}{S_0}\mathbb{E}\left(S_t-S_0\E^{xt}\right)^{+}.$$
Let $\delta>0$ 
and define 
$\epsilon(\delta):=\left(4x^2/\hat{\sigma}^4_{\infty}(x)-1\right)\delta/16>0$.
The reason for this definition lies in 
the following identity
$$A_{\mathrm{BS}}(x,\hat{\sigma}_{\infty}(x),\hat{a}_1(x))\E^{\pm\epsilon(\delta)}= 
A_{\mathrm{BS}}(x,\hat{\sigma}_{\infty}(x),\hat{a}_1(x)\pm\delta)\E^{\mp\epsilon(\delta)}\quad\text{for all}\quad x\in\mathbb{R}$$
which holds
by definition~\eqref{eq:DefOfABS}.
By \eqref{eq:eq1} there exist $t^*(\delta)$ and a neighbourhood of 
$x'$ such that for all 
$x$
in this neighbourhood and all
$t>t^*(\delta)$
we have
\begin{align*}
\frac{1}{S_0}\mathbb{E}\left(S_t-S_0\E^{xt}\right)^{+}
& < \frac{A_{\mathrm{BS}}(x,\hat{\sigma}_{\infty}(x),\hat{a}_1(x)+\delta) }{\sqrt{2\pi t}} \E^{-(V_{\mathrm{BS}}^*(x,\hat{\sigma}_{\infty}(x))-x)t}\E^{-\epsilon(\delta)}\nonumber\\
 & \le \frac{1}{S_0}C_{\mathrm{BS}}\left(S_0,S_0 \E^{xt},t,\sqrt{\hat{\sigma}_{\infty}^2(x)+(\hat{a}_1(x)+\delta)/t}\right),\nonumber
\end{align*}
and
\begin{align*}
\frac{1}{S_0}\mathbb{E}\left(S_t-S_0\E^{xt}\right)^{+}
 & >  \frac{A_{\mathrm{BS}}(x,\hat{\sigma}_{\infty}(x),\hat{a}_1(x)-\delta)}{\sqrt{2\pi t}} \E^{-(V_{\mathrm{BS}}^*(x,\hat{\sigma}_{\infty}(x))-x)t}\E^{\epsilon(\delta)}\nonumber\\
 & \ge \frac{1}{S_0}C_{\mathrm{BS}}\left(S_0,S_0\E^{xt},t,\sqrt{\hat{\sigma}_{\infty}^2(x)+(\hat{a}_1(x)-\delta)/t}\right).\nonumber
\end{align*}
By strict monotonicity of the Black-Scholes price as a function of the volatility, we obtain
$$\hat{\sigma}_{\infty}^2(x)+(\hat{a}_1(x)-\delta)/t<\hat{\sigma}^2_t(x) < \hat{\sigma}_{\infty}^2(x)+(\hat{a}_1(x)+\delta)/t,$$
for all 
$x$
in some interval containing 
$x'$
and all 
$t>t^*(\delta)$.
This proves \eqref{eq:ToProveCorrectionTerm}. 
An analogous argument implies the inequality in~\eqref{eq:ToProveCorrectionTerm}
for
$x'\in\mathbb{R}\backslash\{-\theta/2,\bar{\theta}/2\}$
since the functions
$A$
and
$x\mapsto A_{\mathrm{BS}}(x,\hat{\sigma}_{\infty}(x),\hat{a}_1(x))$
are continuous on this set and
by~\eqref{eq:Quadratic}
the following identity holds
$$
\mathcal{I}(x,t;-\theta/2,\bar{\theta}/2)= \mathcal{I}(x,t;-\hat{\sigma}_\infty^2(x)/2,\hat{\sigma}_\infty^2(x)/2)
\quad\text{for all}\quad (x,t)\in\mathbb{R}\times\mathbb{R}_+
$$
(the function
$\mathcal{I}$ is
defined in~\eqref{eq:DefOfI}).
In the cases
$x'\in\{-\theta/2,\bar{\theta}/2\}$
a similar argument proves~\eqref{eq:ToProveCorrectionTerm}
at the point 
$x'$
but not necessarily on its neighbourhood.

All that is left to prove is that the inequality 
in~\eqref{eq:ToProveCorrectionTerm}
holds uniformly on compact subsets of  the complement
$\mathbb{R}\backslash\{-\theta/2,\bar{\theta}/2\}$.
To every point in a compact set we can 
associate a small interval that contains it 
such that inequality~\eqref{eq:ToProveCorrectionTerm}
holds for all $x$ in that interval and all large 
times $t$. This defines a cover of the compact set.
We can therefore find a finite collection of such intervals
that also covers our compact set. It now follows 
that the inequality in~\eqref{eq:ToProveCorrectionTerm}
holds for any $x$ in the compact set 
and all times $t$
that are larger than the maximum of the finite number of
$t^*(\delta)$
that correspond to the intervals in the finite family that
covers the original set.
\end{proof}

\begin{figure}
\begin{center}
\subfigure{\includegraphics[scale=0.3]{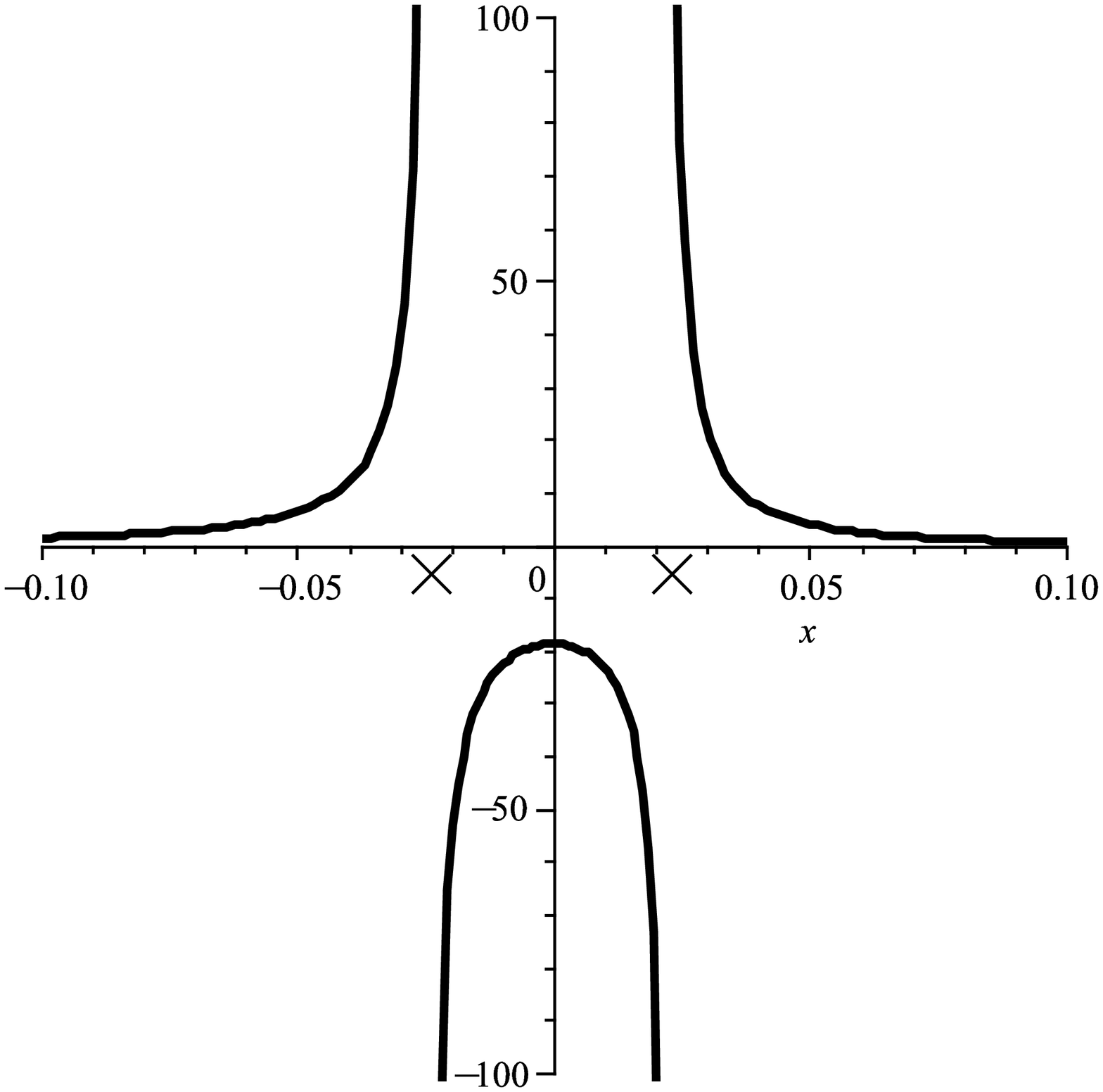}}
\subfigure{\includegraphics[scale=0.3]{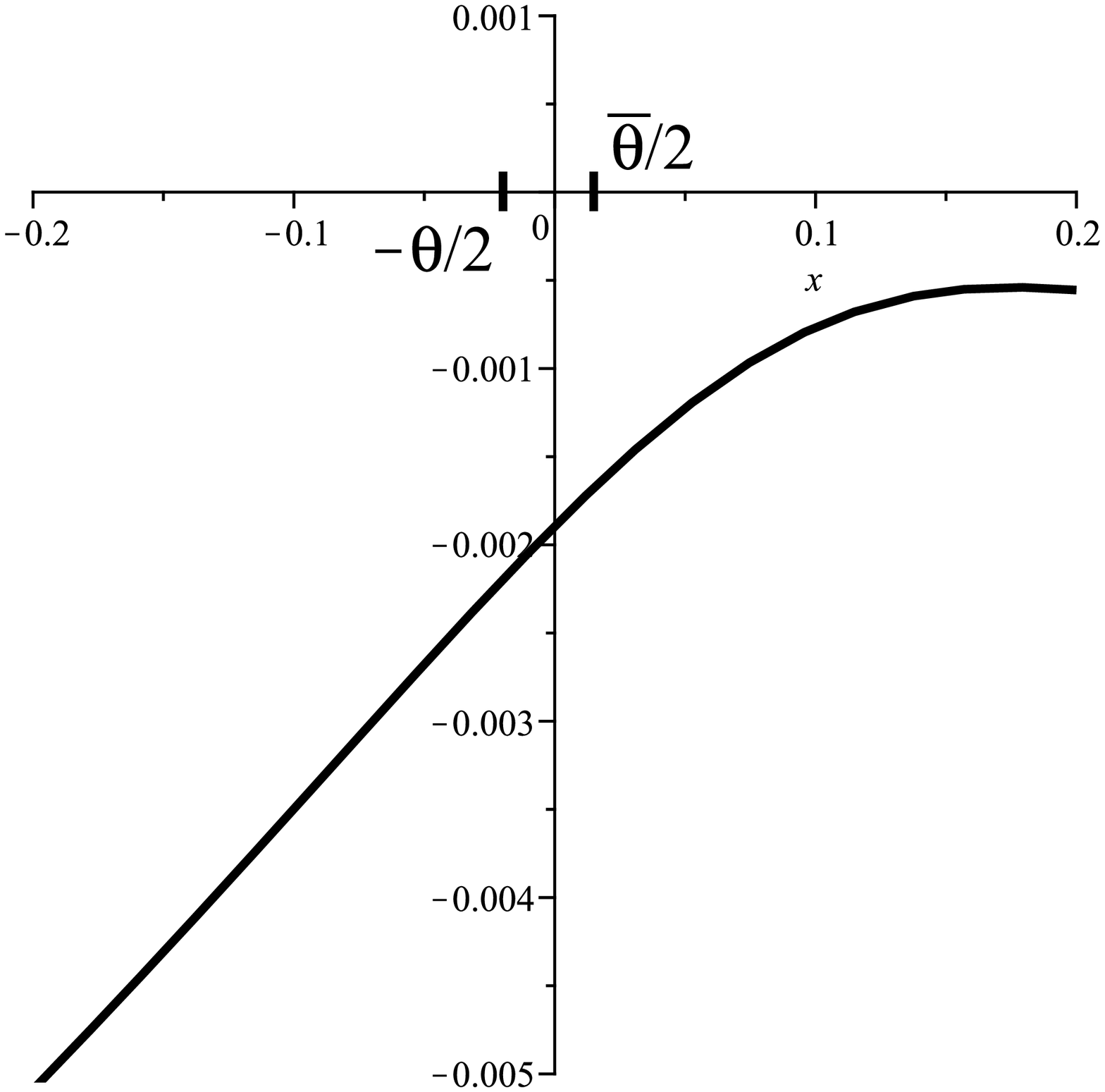}}
\caption{We plot the two functions $A$ (left) and $\hat{a}_1$ (right) with the calibrated parameters on 
page~\pageref{CalibratedParam}. 
The two crosses on the left graph represent the values 
$A(-\theta/2)$
and
$A(\bar{\theta}/2)$
at
$-\theta/2\approx-0.025$ and $\bar{\theta}/2\approx 0.022$
where the function 
$A$ 
is discontinuous.}
\label{FigureAandABS}
\end{center}
\end{figure}
\subsection{The large-time, fixed-strike case}
This section is the translation of Lemmas \ref{lem:LargeTFixedK} and \ref{lem:FixedStrikeOptionBSLemma} in terms of implied volatility asymptotics, and improves the understanding of the behaviour of the Heston implied volatility in the long term. Let $\sigma_t(x)$ denote the implied volatility corresponding to a vanilla call option with maturity $t$ and fixed strike $K=S_0 \exp(x)$ in the Heston model \eqref{eq:HestonModel}. Let us define the function $a_1:\mathbb{R}\to\mathbb{R}$ by
\begin{equation}\label{eq:DefOfa1}
a_1(x):= -8\log\left(-A(0)\sqrt{2V^*(0)}\right)+4\left(2\,p^*(0)-1\right)x,\quad\text{for all }x\in\mathbb{R},
\end{equation}
where $A$ is defined in \eqref{eq:AHeston}, $V^*$ in \eqref{DefOfVStar} and $p^*$ in \eqref{Saddlepoint}. Elementary calculations show that $A(0)<0$. From the properties of $V^*$ on page \pageref{PropertiesOfVStar}, $a_1(x)$ is then well defined as a real number for all $x\in\mathbb{R}$.
\begin{theorem}\label{prop:ImpliedVolFixedStrike}With the assumptions above, we have the following behaviour for the implied volatility in the fixed-strike case
$$\sigma_t^2(x)= 8V^*(0)+a_1(x)/t+o\left(1/t\right),\quad\text{for all }x\in\mathbb{R},\text{ as }t\to\infty,$$
where $V^*$ is given by \eqref{DefOfVStar} and $a_1$ by \eqref{eq:DefOfa1}.
The error term 
$|\sigma_t^2(x)- 8V^*(0)-a_1(x)/t|t$
tends to zero 
as 
$t$
goes to infinity
uniformly on compact subsets of 
$\mathbb{R}$.
\end{theorem}
\begin{proof}
The proof of this theorem is similar to the proof of Theorem \ref{thm:ImpliedVol}.
It is in fact simpler because we do not have to consider special cases for 
$x'\in\mathbb{R}$.
We therefore only give an argument that implies the inequality
\begin{equation}
\label{eq:Ineq_Fixed_Strike}
\sigma_t^2(x)-8 V^*(0)-a_1(x)/t<\delta/t
\end{equation}
holds on some small neighbourhood of 
$x'$
for all large 
$t$.
Lemma \ref{lem:LargeTFixedK} and \eqref{eq:DefOfa1} imply that for all $\epsilon>0$, 
there exists a 
$t^*(\epsilon)$ 
and a neighbourhood of 
$x'$
such that for all 
$x$
in this neighbourhood and all
$t>t^*(\epsilon)$ we have
\begin{align*}
\frac{1}{S_0}\mathbb{E}\left(S_t-S_0 \E^{x}\right)^{+} & <\left(1+(2\pi t)^{-1/2}A(0)\exp\left(x(1-p^*(0))-V^*(0)t\right)\right)\E^{\epsilon}\\
 & =\left(1-\frac{1}{2\sqrt{\pi V^*(0)t}}\exp\left(x/2-a_1(x)/8-V^*(0)t\right)\right)\E^{\epsilon}.
\end{align*}
For any $\delta>0$ a continuity argument implies the existence of 
$\epsilon(\delta)>0$,
which tends to zero uniformly on a neighbourhood 
of
$x'$,
such that
$$\left(1-\frac{1}{2\sqrt{\pi V^*(0)t}}\E^{x/2-a_1(x)/8-V^*(0)t}\right)\E^{\epsilon(\delta)}=\left(1-\frac{1}{2\sqrt{\pi V^*(0)t}}\E^{x/2-\left(a_1(x)+\delta\right)/8-V^*(0)t}\right)\E^{-\epsilon(\delta)}.$$
Therefore there exists $t^*(\delta)>0$ and an interval containing 
$x'$
such that for all
$x$
in this interval and all 
$t>t^*(\delta)$
we have
$$\frac{1}{S_0}\mathbb{E}\left(S_t-\E^x\right)^{+} < \frac{1}{S_0}C_{BS}\left(S_0,S_0\E^x,t,\sqrt{8V^*(0)+(a_1(x)+\delta)/t}\right).$$
The monotonicity of the Black-Scholes formula as a function of the volatility implies~\eqref{eq:Ineq_Fixed_Strike}.
The proof can now be concluded in the same way as in Theorem~\ref{thm:ImpliedVol}.
\end{proof}

\section{Numerical results}\label{section:Numerics}
We present here some numerical evidence of the validity of the asymptotic formula for the implied volatility obtained in Theorem \ref{thm:ImpliedVol}.
We calibrated the Heston model on the European vanilla options on the Eurostoxx 50 on February, 15th, 2006. 
The maturities range from one year to nine years, the strikes from $1460$ up to $7300$, and the initial spot $S_0$ equals $3729.79$. 
The calibration, performed using the Zeliade Quant Framework, by Zeliade Systems, on the whole implied volatility surface gives the following parameters:\label{CalibratedParam}
$\kappa=1.7609$, $\theta=0.0494$, $\sigma=0.4086$, $y_0=0.0464$ and $\rho=-0.5195$.
We then use these calibrated parameters to generate an implied volatility smile for two different maturities, 
$5$ and $9$ years
(note that for convenience the smiles were generated with interest rates and dividends 
equal to zero).
The two plots below contain the generated implied volatility smile for the Heston model 
for the two maturities as well as the zeroth and first order approximations obtained in~\eqref{DefOfSigmaInf} and in 
Theorem~\ref{thm:ImpliedVol}. 
The errors plotted are the differences between the generated Heston implied volatility 
and the implied volatility obtained by the two asymptotic approximations.

\begin{figure}
\begin{center}
\subfigure{\includegraphics[scale=0.3]{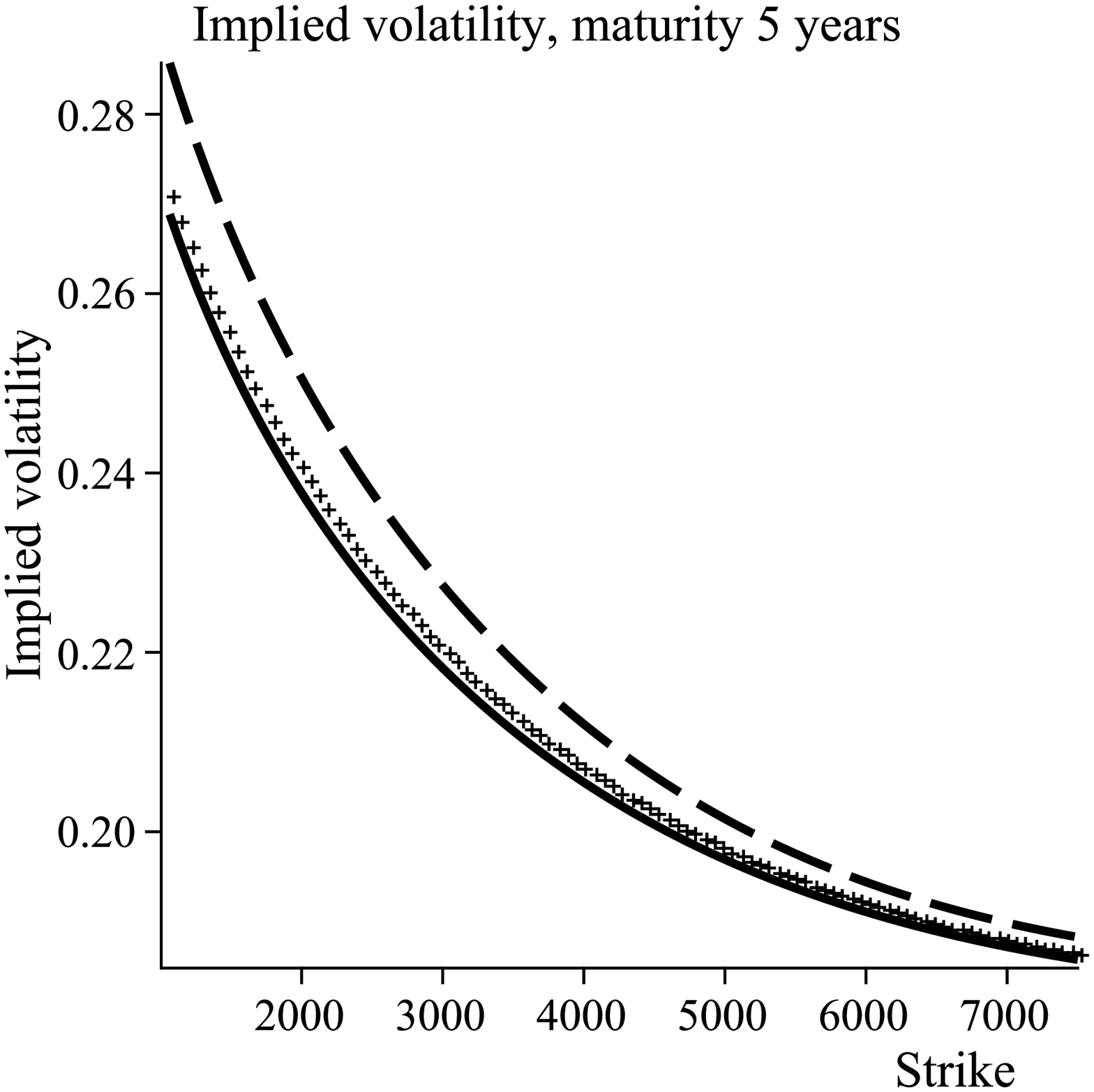}}
\subfigure{\includegraphics[scale=0.3]{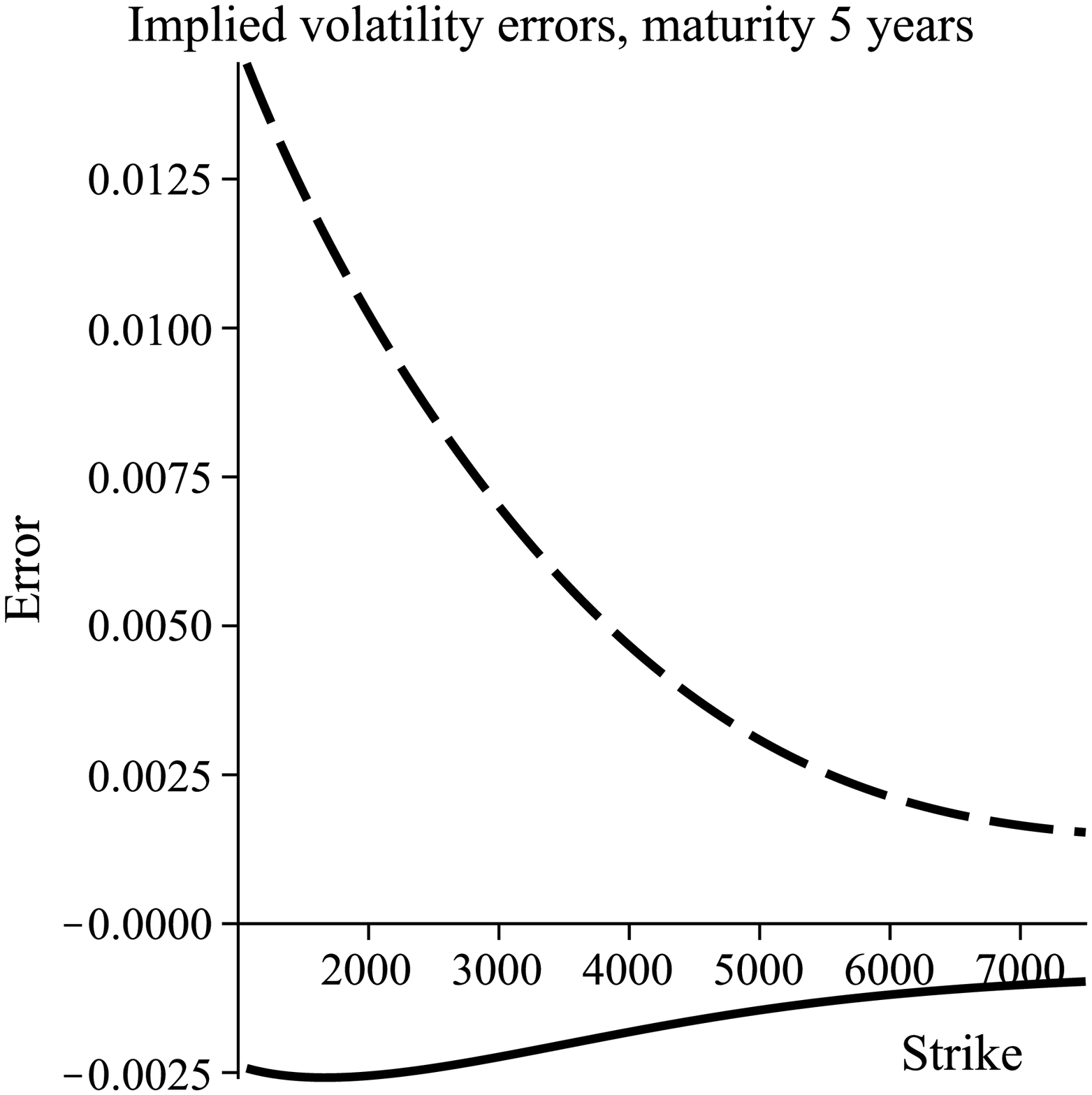}}
\label{fig:FigureSmiles5Years}
\caption{The left plot represents the leading order term $\hat{\sigma}_{\infty}$ (dashed) defined in \eqref{DefOfSigmaInf}, the asymptotic formula in Theorem
\ref{thm:ImpliedVol} (solid) and the true implied volatility (crosses) as functions of the strike $K$ for a maturity equal to $5$ years.
On the right, we plot the corresponding errors between the true implied volatility and $\hat{\sigma}_\infty$ (dashed) and between the true implied volatility and our formula (solid).
The parameter values are given Section~\ref{section:Numerics}.}
\end{center}

\begin{center}
\subfigure{\includegraphics[scale=0.3]{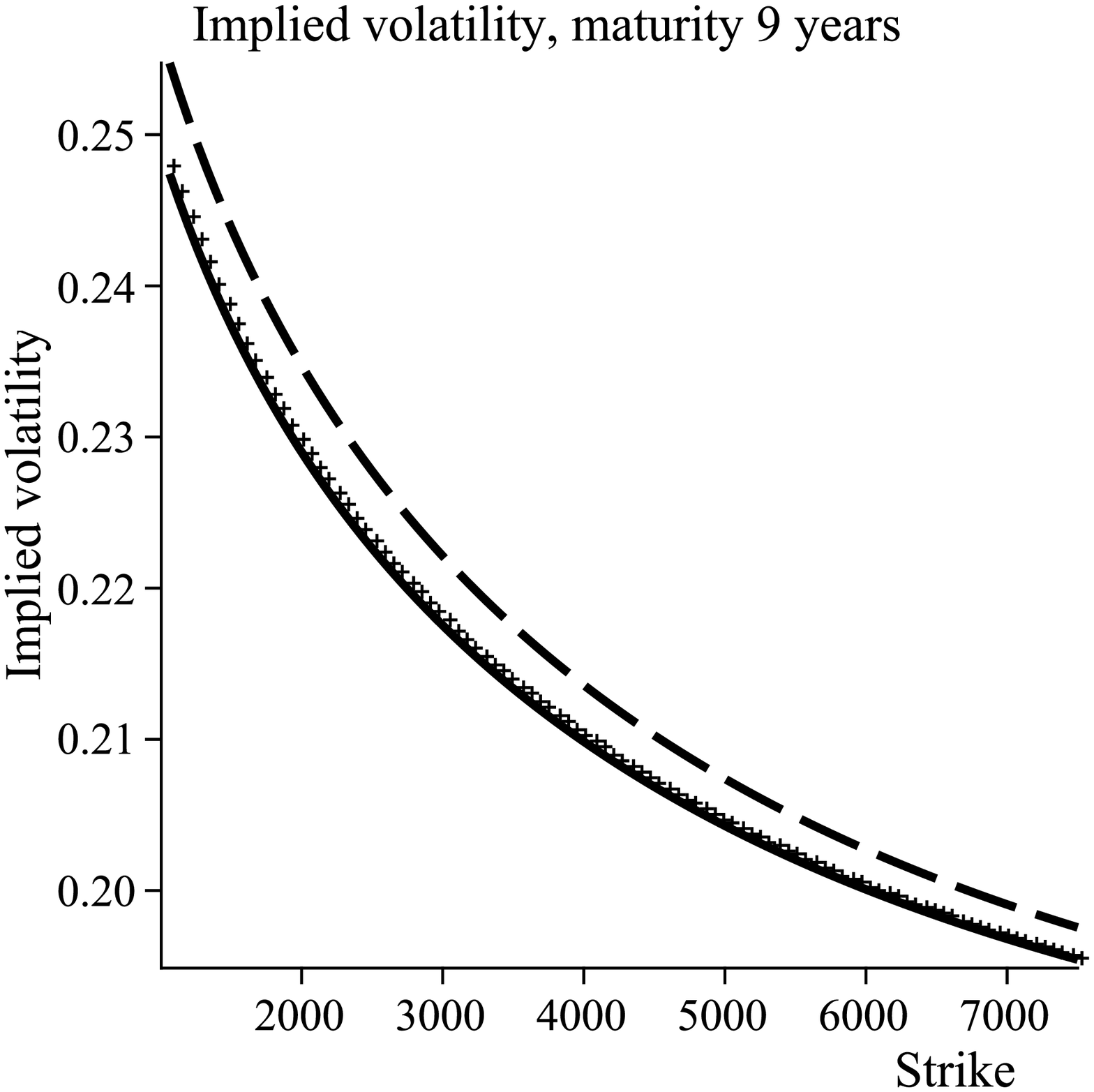}}
\subfigure{\includegraphics[scale=0.3]{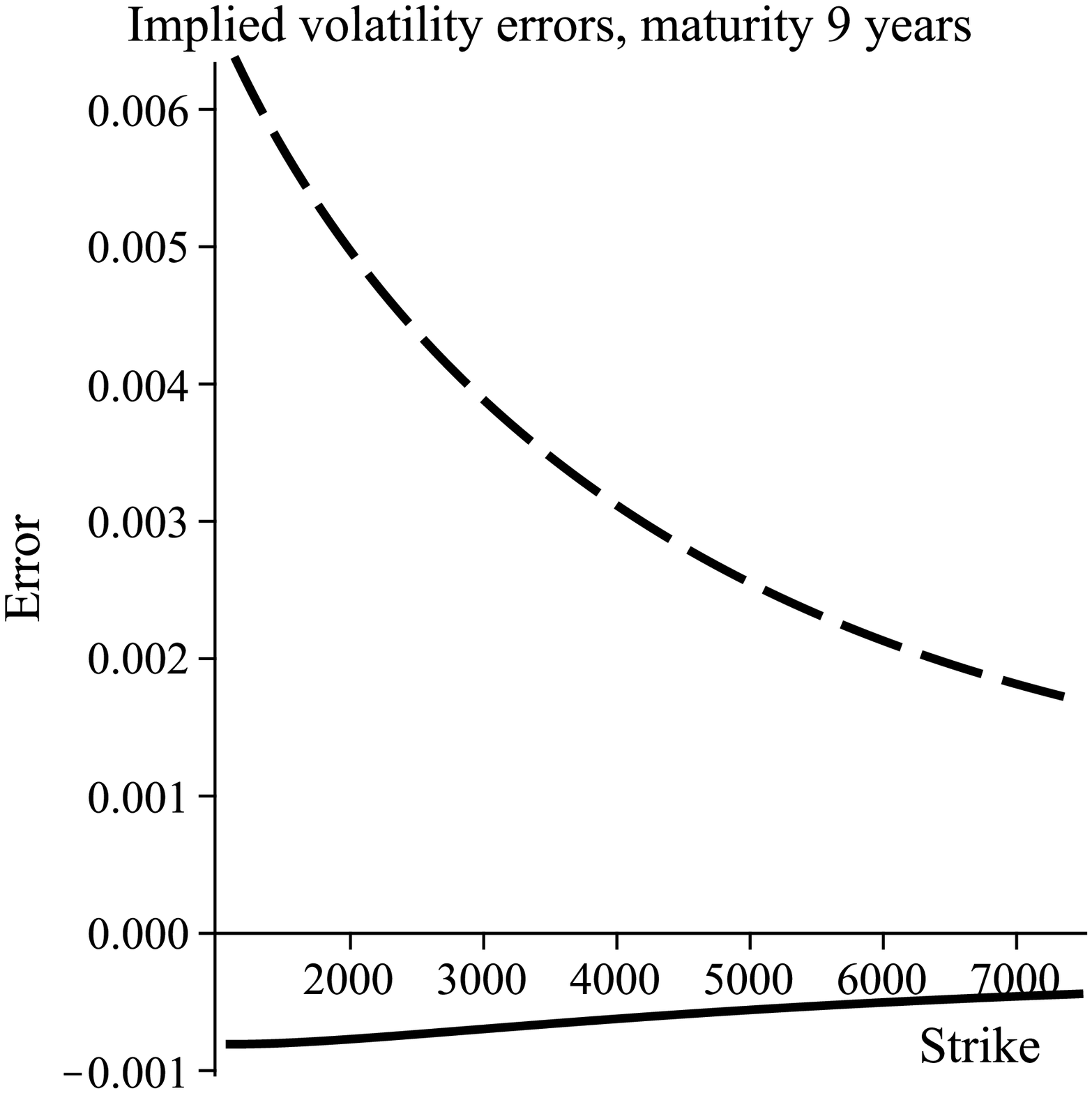}}
\caption{The left plot represents the leading order term $\hat{\sigma}_{\infty}$ (dashed) defined in \eqref{DefOfSigmaInf}, the asymptotic formula in Theorem
\ref{thm:ImpliedVol} (solid) and the true implied volatility (crosses) as functions of the strike $K$ for a maturity equal to $9$ years.
On the right, we plot the corresponding errors between the true implied volatility and $\hat{\sigma}_\infty$ (dashed) and between the true implied volatility and our formula (solid).
The parameter values are given Section~\ref{section:Numerics}.}
\label{fig:FigureSmiles9Years}
\end{center}
\end{figure}

\section{Proof of Theorem \ref{thm:HestonLargeT}}\label{ProofMainTheorem}
The proof of the theorem is divided into a series of steps: we first write the Heston 
call price in terms of an inverse Fourier transform of the characteristic function of 
the stock price \eqref{eq:I}.  Then we prove a large-time estimate for the characteristic 
function (Lemma \ref{lem:AsymptoticBehaviourPhi}). The next step is to deform the contour 
of integration of the inverse Fourier transform through the saddlepoint of the integrand 
(Equation \eqref{Saddlepoint} and Proposition \ref{PropIndicators}). Finally, studying 
the behaviour of the integral around this saddlepoint (Proposition~\ref{prop:MainProp}) 
and bounding the remaining terms (Lemma \ref{LemmaTailEstimates}) completes the proof. 
The special cases
$x=-\theta/2$
and
$x=\bar{\theta}/2$
in formula~\eqref{eq:AHeston}
are proved in 
Sections~\ref{sec:Construction} and~\ref{SubsectionProofSpecial}.

\subsection{The Lee-Fourier inversion formula for call options}\label{sub:LeeFourier}
Using similar notation to Lee \cite{Lee04}, set
$$A_{t,X}:=\left\{\nu \in \mathbb{R}: \mathbb{E}\left(\exp\Big(\nu\left(X_t-x_0\right)\Big)\right)<\infty\right\},\quad\text{for all }t\geq 0,$$
and define the characteristic function $\phi_t:\mathbb{R}\to\mathbb{C}$ of $X_t-x_0$ by
$$\phi_t(z):=\mathbb{E}\left(\exp\Big(\I z(X_t-x_0)\Big)\right),\quad\text{for all }t\geq 0.$$
From Theorem 1 in \cite{LK08} and Proposition 3.1 in \cite{AP07}, we know that $\phi_t(z)$ can be analytically extended for any $z\in\mathbb{C}$ such that $-\Im(z)\in A_{t,X}$, and from our assumptions on the parameters in Section \ref{section:notations}, $\left(p_-,p_+\right)\subseteq A_{t,X}$ for all $t\geq 0$. By Theorem 5.1 in \cite{Lee04}, for any $\alpha\in \left(p_-,p_+\right)$, we have the
following Fourier inversion formula for the price of a call option on $S_t$
\begin{eqnarray*}
\frac{1}{S_0}\mathbb{E}\left(S_t-K\right)^{+} & = & \phi_t(-\I)\ind_{\lbrace 0<\alpha<1\rbrace}+\left(\phi_t(-\I)-\E^x\phi_t(0)\right)\ind_{\lbrace \alpha<0 \rbrace}+\frac{1}{2}\phi_t(-\I)\ind_{\{\alpha=1\}}\\
 & + & \left(\phi_t(-\I)-\frac{\phi_t(0)}{2}\E^{xt}\right)\ind_{\{\alpha=0\}}+\frac{1}{\pi}\int_{\gamma_+}\Re\left(\E^{-\I zx} \,\frac{\phi_t(z-\I)}{\I z-z^2}\right)\D z,
\end{eqnarray*}
where $x:=\log(K/S_0)$ and $\gamma_+:\mathbb{R}\to\mathbb{C}$ is a contour such that $\gamma_+(u):=u-\I(\alpha-1)$. The first four terms on the right hand side are complex residues that arise when we cross the poles of $\left(\I z-z^2\right)^{-1}$ at $z=0$ and $z=\I$. We now set $k=\I-z$, substitute $x$ to  $xt$, and use the fact that $S_t$ is a true martingale for all $t\geq 0$ (see Proposition 2.5 in \cite{AP07}). From now on, as $k$ will always denote a complex number, we use the notation $k=k_r+\I k_i$ for $k_r,k_i\in\mathbb{R}$.
Note that
$$\Re\left(\E^{\I kxt}\frac{\phi_t(-k)}{\I k-k^2}\right)=\mathbb{E}\left(\Re\left(\E^{\I kxt}\frac{\E^{-\I k(X_t-x_0)}}{\I k-k^2}\right)\right),$$
that $k_r\mapsto\Re\left(\E^{\I kxt}\frac{\E^{-\I k(X_t-x_0)}}{\I k-k^2}\right)$ is an even function and $k_r\mapsto\Im\left(\E^{\I kxt}\frac{\E^{-\I k(X_t-x_0)}}{\I k-k^2}\right)$ an odd function. Clearly the normalised call price $S_0^{-1}\mathbb{E}\left(S_t-S_0\exp(xt)\right)^{+}$ is real, so if we take the real part of both sides and break up the integral, we obtain
\begin{align}\label{eq:I}
\frac{1}{S_0}\mathbb{E}\left(S_t-S_0\E^{xt}\right)^{+} & = \ind_{\lbrace 0<\alpha<1\rbrace}+\left(1-\E^{xt}\right)\ind_{\lbrace \alpha<0\rbrace}+ \frac{1}{2}\ind_{\{\alpha=1\}}+\left(1-\frac{1}{2}\E^{xt}\right)\ind_{\{\alpha=0\}}\nonumber\\
 & + \frac{\exp(xt)}{2\pi}\Re\left(\left(\int_{\gamma_\alpha}+\int_{\zeta_\alpha}\right)
\E^{\I kxt}\frac{\phi_t(-k)}{\I k-k^2}\D k\right),
\end{align}
for any $R>0$, where, for any $\alpha\in\mathbb{R}$, we define the contours \begin{equation}\label{ContourGamma}
\gamma_\alpha:(-\infty,-R]\cup[R,+\infty)\to\mathbb{C}\text{ such that }\gamma_\alpha(u):=u+\I\alpha,
\end{equation}
and 
\begin{equation}\label{ContourZeta}
\zeta_\alpha:(-R,R)\to\mathbb{C}\text{ such that }\zeta_\alpha(u):=u+\I\alpha.
\end{equation}
For ease of notation, we do not write explicitly the dependence of these contours on $R$. We will see later how to choose $R$. In the following lemma, we characterise the large-time asymptotic behaviour of the characteristic function $\phi_t$.
\begin{lemma}\label{lem:AsymptoticBehaviourPhi}
For all $k\in\mathbb{C}$ such that $-k_i\in\left(p_-,p_+\right)$, we have
$$\phi_t(k)=\exp\Big(V(\I k)t\Big)U(\I k)\left(1+\epsilon(k,t)\right),\quad\text{as }t\to\infty,$$
$\Re(d(k))>0$, and $\epsilon(k,t)=O\left(\E^{-t\Re(d(k))}\right)$, where $U$ is defined in \eqref{eq:DefOfU}, $p_-$, $p_+$ in \eqref{DefOfPpm}, $d$ in \eqref{DefinitionOfD} and $V$ is the analytic continuation of formula \eqref{eq:V(p)}.
\end{lemma}
\noindent If $-k_i$ is not in $\left(p_-,p_+\right)$, this large time behaviour of $\phi_t$ still holds, but $\Re(d(k))$ might be null (for instance if $k_r=0$) so that $\epsilon(k,t)$ does not tend to $0$ as $t\to\infty$.
\begin{proof}
From \cite{AMST06}, we have, for all $k\in\mathbb{C}$ such that $-k_i\in\left(p_-,p_+\right)$,
\begin{align}
\label{eq:CF}
\phi_t(k) = & \exp\left(V(\I k)t-\frac{2\kappa\theta}{\sigma^2}\log\left(\frac{1-g(k)\E^{-d(k)t}}{1-g(k)}\right)\right)\exp\left(\frac{y_0}{\kappa\theta}V(\I k)\frac{1-\E^{-d(k)t}}{1-g(k)\E^{-d(k)t}}\right),
\end{align}
where $d$ is defined in \eqref{DefinitionOfD}, $V$ is the analytic extension of formula \eqref{eq:V(p)} and the correct branch for the complex logarithm and the complex square root function is the principal branch (see also \cite{AMST06} and \cite{LK08}) and $g:\mathbb{C}\to\mathbb{C}$ is defined by
\begin{equation}\label{eq:DefOfG}
g(k):= \frac{\kappa-\I \rho\sigma k -d(k)}{\kappa-\I \rho\sigma k+d(k)},\quad\text{for all }k\in\mathbb{C}.
\end{equation}
For all $k\in\mathbb{C}$ such that $-k_i\in\left(p_-,p_+\right)$, we have $\Re(d(k))>0$. Let 
$$\epsilon_1(k,t):=\left(1-g(k)\E^{-d(k)t}\right)^{-2\kappa\theta/\sigma^2} \text{ and }\epsilon_2(k,t):=\exp\left\{-\frac{2d(k)V(\I k)y_0}{\kappa\theta\left(\kappa-\rho\sigma \I k+d(k)\right)}\left(\E^{d(k)t}-g(k)\right)^{-1}\right\}.$$
Then we have
$$\phi_t(k)=\exp\Big(V(\I k)t\Big)U(\I k)\epsilon_1(k,t)\epsilon_2(k,t),\quad\text{for all }t\geq 0.$$
Then, as $t$ tends to infinity we have
$$\epsilon_{1}(k,t)= 1+\frac{2\kappa\theta g}{\sigma^2}\E^{-d(k)t}+O\left(\E^{-2d(k)t}\right)\text{ and }\epsilon_{2}(k,t)= 1+\frac{c}{\E^{d(k)t}-1}+O\left(\left(\E^{d(k)t}-1\right)^{-2}\right)$$
for some constant $c$. Set $\epsilon(k,t):=\epsilon_1(k,t)\epsilon_2(k,t)-1$ and the lemma follows.
\end{proof}

\subsection{The saddlepoint and its properties}
We first recall the definition of a saddlepoint in the complex plane (see \cite{Bleistein}):
\begin{definition}
Let $F:\mathcal{Z}\to\mathbb{C}$ be an analytic complex function on an open set $\mathcal{Z}$. A point $z_0\in\mathcal{Z}$ such that the complex derivative $\frac{\D F}{\D z}$ vanishes is called a saddlepoint.
\end{definition}
\noindent Note that the function $V:\left(p_-,p_+\right)\to\mathbb{R}$ defined in \eqref{DefOfFunctionV} can be analytically extended and we define 
for every 
$x\in\mathbb{R}$
the function $F_x:\mathcal{Z}\to\mathbb{C}$ by
\begin{equation}
\label{eq:F}
F_x(k):=-\I kx-V(-\I k),\quad\text{where }\mathcal{Z}:=\left\{k\in\mathbb{C}:k_i\in\left(p_-,p_+\right)\right\}.
\end{equation}
Note that the exponent of the integrand in \eqref{eq:I} has the form $-F_x(k)t$ by Lemma \ref{lem:AsymptoticBehaviourPhi}, therefore the saddlepoint properties of $F_x$ given in the following elementary lemma are fundamental.
\begin{lemma}\label{lemma:saddlepointz}
The saddlepoints of the complex function $F_x:\mathcal{Z}\to\mathbb{C}$ are given by
$$z_0^{\pm}(x)=\I\frac{\sigma-2\kappa\rho\pm(\kappa\theta\rho+x\sigma)\eta(x^2\sigma^2+2x\kappa\theta\rho\sigma+\kappa^2\theta^2)^{-1/2}}{2\sigma\bar{\rho}^2}\in\mathcal{Z},$$
where $\eta:=\sqrt{\sigma^2+4\kappa^2-4\kappa \rho\sigma}$.
\end{lemma}
\begin{proof}
Since we are looking for the saddlepoint in $\mathcal{Z}$, we can use the representation in \eqref{eq:V(p)} for the function $V$. Therefore the equation
$F_x'(z)=0$ is quadratic and hence has the two purely imaginary solutions 
$z_0^{\pm}(x)$ since the expression 
$(x^2\sigma^2+2x\kappa\theta\rho\sigma+\kappa^2\theta^2)=\left(x\sigma+\kappa\theta\rho\right)^2+\kappa^2\theta^2(1-\rho^2)$ 
is
strictly positive for any 
$x\in\mathbb{R}$. It is also clear from the definition of $p_-$ and $p_+$ given in \eqref{DefOfPpm} and the assumptions on the coefficients
that $\Im\left(z_0^\pm(x)\right)\in\left(p_-,p_+\right)$, and therefore $z_0^\pm(x)\in\mathcal{Z}$ are saddlepoints of $F_x$.
\end{proof}
\noindent The next task is to choose the saddlepoint of the function $F_x$ in such a way that it converges to the 
saddlepoint of the function 
$F_x^{\mathrm{BS}}(k):=-\I kx-V_{BS}(-\I k)$ for all $k\in\mathcal{Z}$, where $V_{\mathrm{BS}}$ is given by \eqref{eq:DefOfVbs}, in the Black-Scholes model as both the
volatility of volatility and the correlation in model \eqref{eq:HestonModel} tend to zero. It is easy to see that the saddlepoint of 
$F_x^{\mathrm{BS}}$ equals $\I p^*_{\mathrm{BS}}(x)$ for any $x\in\mathbb{R}$, were $p^*_{\mathrm{BS}}$ is given by \eqref{eq:DefOfpStarBS}. We can rewrite
$\Im\left(z_0^\pm(x)\right)$ defined in Lemma \ref{lemma:saddlepointz} as
\begin{equation}\label{eq:ConvPbsToP}
\Im\left(z_0^\pm(x)\right)=\frac{1}{2\bar{\rho}^2}-\frac{\kappa\rho}{\sigma\bar{\rho}^2}\pm\frac{\kappa\theta\rho\eta(x^2\sigma^2+2x\kappa\theta\rho\sigma+\kappa^2\theta^2)^{-1/2}}{2\sigma\bar{\rho}^2}\pm\frac{x\sigma\eta(x^2\sigma^2+2x\kappa\theta\rho\sigma+\kappa^2\theta^2)^{-1/2}}{2\sigma\bar{\rho}^2},
\end{equation}
where $\bar{\rho}$ is defined page \pageref{defofrhobar}. The first term converges to $1/2$ and the last one to $\pm x/\theta$ as $(\rho,\sigma)$ tends to $0$. When both $\rho$ and $\sigma$ tend to $0$, a Taylor expansion at first order of the third term gives
\begin{equation}
\label{eq:Limit_RhoSigma}
\frac{\kappa\theta\rho\eta(x^2\sigma^2+2x\kappa\theta\rho\sigma+\kappa^2\theta^2)^{-1/2}}{2\sigma\bar{\rho}^2}=\frac{\rho\eta}{2\sigma\bar{\rho}^2}.
\end{equation}
Take now the positive sign in \eqref{eq:ConvPbsToP}, then the second and third terms cancel out in the limit because 
$\eta$ 
converges to 
$2\kappa$ 
as 
$(\rho,\sigma)$ 
tends to $0$. In that case we have 
$\lim\limits_{(\rho,\sigma)\to 0} \Im\left(z_0^+(x)\right)=1/2+x/\theta=p_{\mathrm{BS}}^*(x)$ 
for all $x\in\mathbb{R}$, where the Black-Scholes variance is equal to $\theta$. 
If we were to take 
$\Im\left(z_0^-(x)\right)$
for the saddlepoint, in the limit we would not recover 
$p_{\mathrm{BS}}^*(x)$
by~\eqref{eq:Limit_RhoSigma}, since the function $(\rho,\sigma)\mapsto \rho/\sigma$ has no limit as the pair $(\rho,\sigma)$ tends to $0$.
Therefore we define the saddlepoint to be
$z_0^+(x)$.
Moreover we observe the following equality
\begin{equation*}
\Im(z_0^+(x))
=p^*(x),\quad\text{for all }x\in\mathbb{R},
\end{equation*}
where 
$p^*(x)$
is defined in~\eqref{Saddlepoint}.
\begin{remark}
In the fixed strike case, i.e. when $x=0$, we obtain
$$p^*(0)=\frac{-2\kappa\rho+\sigma+\rho \eta}{2\sigma\bar{\rho}^2}.$$
The corresponding saddlepoint $\I p^*(0)$ is the same as the one in Chapter 6 in \cite{Lewis00}.
\end{remark}
The following lemma is of fundamental importance and will be the key tool for Proposition \ref{prop:MainProp}.
\begin{lemma}
\label{lem:FAlongHorizontalContour}
Let $k\in\mathcal{Z}$. Then, for any 
$k_i\in\left(p_-,p_+\right)$, 
the function 
$k_r\mapsto\Re\left(-\I kx-V(-\I k)\right)$ has a unique minimum at $0$
and is strictly decreasing
(resp. increasing)
for 
$k_r\in(-\infty,0)$
(resp. 
$k_r\in(0,\infty)$).
\end{lemma}
\begin{proof}
Note that the statement in the lemma
is equivalent to the map 
$k_r\mapsto-\Re\left(V(-\I(k_r+\I k_i))\right)$ having a unique minimum at $k_r=0$ for any $k_i\in\left(p_-,p_+\right)$
and being increasing (resp. decreasing) on the positive
(resp. negative) halfline. 
Let $k_i\in\left(p_-,p_+\right)$, then
$$\Re\Big(V(-\I(k_r+\I k_i))\Big) = \frac{\kappa \theta}{\sigma^2} \left(\kappa-\rho \sigma k_i-\Re\left(\sqrt{u(k_r)+\I v(k_r)}\right)\right),$$
where
$$u(k_r) := \sigma^2\bar{\rho}^2k_r^2-\sigma^2\bar{\rho}^2 k_i^2-\sigma(2\kappa\rho-\sigma)k_i+\kappa^2\quad\text{ and }\quad v(k_r) := \left(2\kappa\rho-\sigma+2\sigma\bar{\rho}^2 k_i\right)\sigma k_r.$$
From the identity and the fact that the principal value of the square-root is used, we get
\begin{equation}\label{RealPartSqrt}
\Re\left(\sqrt{u(k_r)+\I v(k_r)}\right)=\frac{1}{2}\sqrt{2u(k_r)+2\sqrt{u^2(k_r)+v^2(k_r)}}
\end{equation}
is monotonically increasing in $u$, $u^2$ and $v^2$. First, note
that $u'(k_r)=2\sigma^2 \bar{\rho}^2 k_r$, hence $u$ is a parabola with a unique minimum at $k_r=0$, so that, from \eqref{RealPartSqrt}, it suffices to prove the following claim:\\
\textbf{Claim:} For every $k_i\in\left(p_-,p_+\right)$, the function $g:=u^2+v^2$ has a unique (strictly positive) minimum attained at $k_r=0$ and is strictly increasing (resp. decreasing) for $k_r>0$ (resp. $k_r<0$).\\
Let us write $u(k_r)=\sigma^2\bar{\rho}^2k_r^2+\psi(k_i)$, for all $k_r\in\mathbb{R}$, where $\psi(k_i):=\kappa^2-\sigma^2\bar{\rho}^2 k_i^2-\sigma(2\kappa\rho-\sigma)k_i$. We have
\begin{equation}\label{Eq:Polynom}
g(k_r)=\sigma^4\bar{\rho}^4k_r^4+\left(2\bar{\rho}^2\psi(k_i)+\left(2\kappa\rho-\sigma+2\sigma\bar{\rho}^2k_i\right)^2\right)\sigma^2k_r^2+\psi(k_i)^2\quad\text{for all }k_r\in\mathbb{R}.
\end{equation}
The coefficient $\sigma^2\bar{\rho}^4$ and the constant $\kappa^2$ are strictly positive, 
so the claim follows if $\chi(k_i)>0$ for all $k_i\in(p_-,p_+)$, where
$$\chi(k_i):=\left(2\bar{\rho}^2\psi(k_i)+\left(2\kappa\rho-\sigma+2\sigma\bar{\rho}^2k_i\right)^2\right)=
2\sigma^2\bar{\rho}^4k_i^2 + 2\sigma\bar{\rho}^2\left(2\kappa\rho-\sigma\right)k_i+\kappa^2+\left(2\kappa\rho-\sigma\right)^2.$$
The discriminant is $\Delta_{\chi}=-4\sigma^2\bar{\rho}^4\left(2\kappa^2+\left(2\kappa\rho-\sigma\right)^2\right)<0$, so that $\chi$ has no real root and is hence always strictly positive. This proves the claim and concludes the proof of the lemma.
\end{proof}
The following two results complete the proof of Theorem \ref{thm:HestonLargeT}, by studying the behaviour of the two integrals in \eqref{eq:I} as the time to maturity tends to infinity. The following lemma proves that the integral along $\gamma_{p^*(x)}$ is negligible and Proposition \ref{prop:MainProp} hereafter provides the asymptotic behaviour of the integral along the contour $\zeta_{p^*(x)}$.
\begin{lemma}\label{LemmaTailEstimates}
For any $x\in\mathbb{R}$ and any $m>V^*(x)$, there exists $R(m)>0$ such that for every $k\in\mathcal{Z}$ with $|k_r|>R(m)$, we have
\begin{equation}\label{eq:InequalityModulus}
\left|\exp\left(\I kxt\right)\phi_t(-k)\right|\leq\exp(-mt),\quad\text{for all }t\geq 1.
\end{equation}
Therefore
\begin{equation}\label{eq:ModulusIntegral}
\left|\E^{xt}\int_{\gamma_{p^*(x)}}\E^{\I kxt}\frac{\phi_t(-k)}{\I k-k^2}\D k\right|=O\left(\E^{-(m-x)t}\right),
\end{equation}
where the contour $\gamma_{p^*(x)}$ is defined in \eqref{ContourGamma}.
\end{lemma}
\begin{remark}
(i) For every $x\in\mathbb{R}$, we have $V^*(x)\geq x$ by (d) on page \pageref{PropertiesOfVStar} and hence $m-x>0$. Therefore the modulus of the integral \eqref{eq:ModulusIntegral} tends to zero exponentially in time and in $m$.\\
(ii) Recall that $p^*(x)\in(p_-,p_+)$ by Proposition \ref{PropIndicators} and hence inequality \eqref{eq:InequalityModulus} can be applied when estimating integral \eqref{eq:ModulusIntegral}.
\end{remark}
\begin{proof}
We only need to prove \eqref{eq:ModulusIntegral}. Recall from Lemma \ref{lem:AsymptoticBehaviourPhi}, after some rearrangements, that
$$\phi_t(-k)=\E^{\left(t+y_0/(\kappa\theta)\right)V(-\I k)}\left(1-g(-k)\right)^{2\kappa\theta/\sigma^2}\left(1+O\left(\E^{-t\Re(d(-k))}\right)\right).$$
It follows from equations \eqref{DefinitionOfD}, \eqref{eq:V(p)} and \eqref{eq:DefOfG} that
\begin{align*}
\Re\left(d\left(-(k_r+\I k_i)\right)\right) & \sim \sigma\bar{\rho}\left|k_r\right|,\quad\text{ as } \left|k_r\right|\to\infty,\\
V\left(-\I(k_r+\I k_i)\right) & \sim -\kappa\theta\bar{\rho}\left|k_r\right|/\sigma,\quad\text{ as } \left|k_r\right|\to\infty,\\
\lim\limits_{\left|k_r\right|\to\infty} g\left(-(k_r+\I k_i)\right) & = \left(\rho-\I\bar{\rho}\right)^2\ne 1,\quad\text{ since } \left|\rho\right|<1.
\end{align*}
Hence there exists a constant $C>0$, 
independent of 
$k$
and
$t$,
such that 
the following inequality holds
$$\left|\E^{\I kxt}\phi_t(-k)\right|\leq 
C\exp\Big(-t\left(k_i x -1+\sigma^{-1}\kappa\theta\bar{\rho}\left|k_r\right|\right)\Big).$$
Define $R(m):=\max\{\sigma\left(m+1-k_ix+ \log(C)\right)/(\kappa\theta\bar{\rho}),1\}$. 
Then if $|k_r|>\max\left\{R(m),R\right\}$,
where the positive constant
$R$
is given in definition~\eqref{ContourGamma},
the equality
\eqref{eq:ModulusIntegral} follows.
\end{proof}

\begin{proposition}\label{prop:MainProp}
For any $R>0$ and $x\in\mathbb{R}\setminus\left\{-\theta/2,\bar{\theta}/2\right\}$, we have as $t\to\infty$,
\begin{equation}\label{eq:IntegralI}
\frac{\exp(xt)}{2\pi}\Re\left(\int_{\zeta_{p^*(x)}}
\E^{\I kxt}\frac{\phi_t(-k)}{\I k-k^2}\D k\right)=\frac{\exp\left(-\left(V^*(x)-x\right)t\right)}{\sqrt{2\pi t}}\left(A(x)+O\left(1/t\right)\right),
\end{equation}
where $A$ is given in \eqref{eq:AHeston}, $V^*$ in \eqref{DefOfVStar} and $\zeta_{p^*(x)}$ in \eqref{ContourZeta}.
\end{proposition}

\begin{proof}
Let $x\in\mathbb{R}\setminus\left\{-\theta/2,\bar{\theta}/2\right\}$. Applying Lemma \ref{lem:AsymptoticBehaviourPhi} on the compact interval $[-R,R]$, we have
$$\int_{\zeta_{p^*(x)}}\E^{\I kxt} \frac{\phi_t(-k)}{\I k-k^2}\D k=\int_{\zeta_{p^*(x)}}\frac{U\left(-\I k\right)}{\I k-k^2}\E^{(\I kx +V(-\I k))t}\left(1+\epsilon(k,t)\right)\D k,$$
for $t$ large enough. 
By Lemma \ref{lem:FAlongHorizontalContour}, 
we know that 
$k_r\mapsto -\Re(\I (k_r+\I p^*(x)) x +V(-\I(k_r+\I p^*(x))))$
has a unique minimum at 
$k_r=0$
and
the value of the function at this minimum equals
$V^*(x)$
by the definition of 
$V^*$.
The functions 
$V$ and $U$ are analytic along the contour of integration and
thus, by Theorem 7.1, section 7, chapter 4 in \cite{Olv74}, we have 
\begin{align*}
\Re\left(\int_{\zeta_{p^*(x)}}\frac{U(-\I k)}{\I k-k^2}\E^{(\I kx +V(-\I k))t}\D k\right) & = \frac{\E^{xt}}{\sqrt{\pi t}}\E^{-V^*(x)t}\left(\frac{U(p^*(x))}{\sqrt{2V''(p^*(x))}}+O(1/t)\right)\\
 & =\frac{\exp(-(V^*(x)-x)t)}{\sqrt{2\pi t}}\left(A(x)+O(1/t)\right)
\end{align*}
as $t$
tends to infinity.
The $\epsilon(k,t)$ term is a higher order term which we can ignore at the level we are interested in.
\end{proof}
Lemma \ref{LemmaTailEstimates} and Proposition \ref{prop:MainProp} complete the proof of Theorem \ref{thm:HestonLargeT} for the general case. Concerning the two special cases, we first introduce a new contour, the path of steepest descent, which represents the optimal (in a sense made precise below) path of integration. Note that, the general case can also be proved using this path, but Lemma \ref{lem:FAlongHorizontalContour} simplifies the proof.

\subsection{Construction of the path of steepest descent}
\label{sec:Construction}
We first recall the definition of the path of steepest descent before computing it explicitly for the Heston model in the large time case.\begin{definition}(see~\cite{SS03}) Let $z:=x+\I y,\ x,y\in\mathbb{R}$ and $F:\mathbb{C}\to\mathbb{C}$ be an analytic complex function. The steepest descent contour $\gamma:\mathbb{R}\to\mathbb{C}$ is a map such that
\begin{itemize}
\item $\Re(F)$ has a minimum at some point $z_0\in\gamma$ and $\Re(F''(z_0))>0$ along $\gamma$.
\item $\Im(F)$ is constant along $\gamma$.
\end{itemize}
These two conditions imply that $F'(z_0)=0$.
\end{definition}
\noindent The following lemma computes the path of steepest descent explicitly in the Heston case for the function $F_x$ given in \eqref{eq:F} passing through $\I p^*(x)$.
\begin{lemma}\label{lem:LemmaContour}
The path of steepest descent $\gamma$ in the Heston model is the map $\gamma:\mathbb{R}\to\mathbb{C}$ defined by
$$\gamma(s):=s+\I k_i(s),\quad\text{for all }s\in\mathbb{R},$$
where
\begin{equation}\label{eq:kiofs}
k_i(s):=-\frac{\beta-(\kappa\theta\rho+x\sigma)\sqrt{\psi(s)}}{2\kappa\theta\sigma\xi\bar{\rho}^2},
\end{equation}
$\beta:=\kappa\theta\xi(2\kappa\rho-\sigma),\ \psi(s):=4\sigma^2\bar{\rho}^2\xi^2 s^2+\kappa^2\theta^2\left((2\kappa\rho-\sigma)^2+4\kappa^2\bar{\rho}^2\right)\xi\text{ and }\xi:=(\kappa\theta\rho+x\sigma)^2+\kappa^2\theta^2\bar{\rho}^2.$
\end{lemma}
\noindent Note that $\xi$ is strictly positive, so that the function $k_i$ is well defined.
\begin{proof}
By definition, the contour of steepest descent is such that the function $\Im\left(F_x\circ\gamma\right)$ remains constant. So we look for the map $\gamma$ such that
$\Im(F_x(\gamma(s)))=0$, for all $s\in\mathbb{R}$ because $F_x(\gamma(0))$ is already real.  Using the identity $\Im\left(\sqrt{x+\I y}\right)=4\left(2x+2\sqrt{x^2+y^2}\right)^{-1/2}$,
for all $x,y\in\mathbb{R}$, we find that the function $F_x\circ\gamma$ is real along the contour $\gamma:s\mapsto s+\I k_i(s)$. Note also that this contour is orthogonal to the imaginary axis at $k_i(0)$ (see Exercise 2, Chapter 8 in \cite{SS03}).
\end{proof}
\begin{remark}\text{}
\begin{itemize}
\item The contour $\gamma$ depends on $x$, but for clarity we do not write this dependence explicitly.
\item The construction of $\gamma$ is such that the saddlepoint defined in \eqref{Saddlepoint} satisfies $\I p^*(x)=\gamma(0)$.
\item We have  $k_i(s)=\Im(\gamma(s))$ is an even function of $s$, i.e. $\gamma$ is symmetric around the imaginary axis.
\end{itemize}
\end{remark}
We now prove Theorem \ref{thm:HestonLargeT} in the two special cases $x\in\left\{-\theta/2,\bar{\theta}/2\right\}$.
In these cases, we need a result similar to Proposition \ref{prop:MainProp}, as Lemma \ref{LemmaTailEstimates} still holds, i.e. we need the asymptotic behaviour of the integral in \eqref{eq:IntegralI} for the two special cases. The problem with these special cases is that $(\I k-k^2)^{-1}$ in the integrand in \eqref{eq:IntegralI} has a pole at the saddlepoint, so we need to deform the contour using Cauchy's integral theorem and take the real part to remove the singularity, before we can use a saddlepoint expansion. 

\subsection{Proof of the call price expansion for the special cases}
\label{SubsectionProofSpecial}
We here prove Theorem \ref{thm:HestonLargeT} in the case $x=\bar{\theta}/2$
for which $p^*\left(\bar{\theta}/2\right)=1$, $V^*\left(\bar{\theta}/2\right)=\bar{\theta}/2$, and for simplicity we also assume that
$\kappa<\left(\sigma-2\rho^2\sigma\right)/(2\rho)$ (the other cases follows similarly). From \eqref{eq:kiofs}, we see that in this case, $\gamma$ lies below the
horizontal contour $\gamma_H:\mathbb{R}\to\mathbb{C}$ such that $\gamma_H(s):=s+\I $ (in the other case, $\gamma$ lies above $\gamma_H$). We want to construct a new
contour leaving the pole outside. Let $\epsilon>0$ and $\gamma_\epsilon:(-\pi,\pi]\to\mathbb{C}$ denote the clockwise oriented circular keyhole contour parameterised as $\gamma_{\epsilon}(\theta):= \I+\epsilon\,\E^{\I\theta}$ around the pole. To leave the pole outside the new contour of integration, we need to follow $\gamma$ on $\mathbb{R}_-$, switch to the keyhole contour as soon as we touch it, follow it clockwise (above the pole), and get back to $\gamma$ on $\mathbb{R}_+$. As $\gamma$ is below $\gamma_H$, it intersects $\gamma_\epsilon$ on its lower half, which can be analytically represented as $\gamma_\epsilon^-:[-\epsilon,\epsilon]\to\mathbb{C}$ such that $\gamma_{\epsilon}^-(s):=\I+s-\I\sqrt{\epsilon^2-s^2}$. From \eqref{eq:kiofs}, the two contours intersect at $s^*=\pm\epsilon$. Choose now $0<\epsilon<\delta<R$ (Lemma \ref{LemmaSpecialCase} makes the choice of $\delta$ precise and $\epsilon$ must be such that $1+2\epsilon<p_+$) and define the following contours
(they are all considered anticlockwise, see Figure \ref{fig:keyhole})
\begin{itemize}
\item $\gamma_{\delta,R}:[-R,-\delta]\cup[\delta,R]\to\mathbb{C}$ given by $\gamma_{\delta,R}(u)=u+\I k_i(\delta)$, with $k_i(\delta)$ 
defined in~\eqref{eq:kiofs};
\item $\gamma_{\epsilon,\delta}$ is the restriction of $\gamma$ to the union of the intervals $[-\delta,-\epsilon] \cup [\epsilon,\delta]$;
\item $\gamma_\epsilon^U$ is the portion of the circular keyhole contour $\gamma_{\epsilon}$ which lies above $\gamma$, i.e. the upper half keyhole contour as well as the two sections of $\gamma_\epsilon$ between $\gamma$ and $\gamma_H$;
\item $\Gamma_{R,\epsilon,\delta}^{\pm}$ are the two vertical strips joining $\pm R+\I(1+2\epsilon)$ to $\pm R-\I k_i(\delta)$.
\end{itemize}
\begin{figure}
\centering
\includegraphics[scale=0.4]{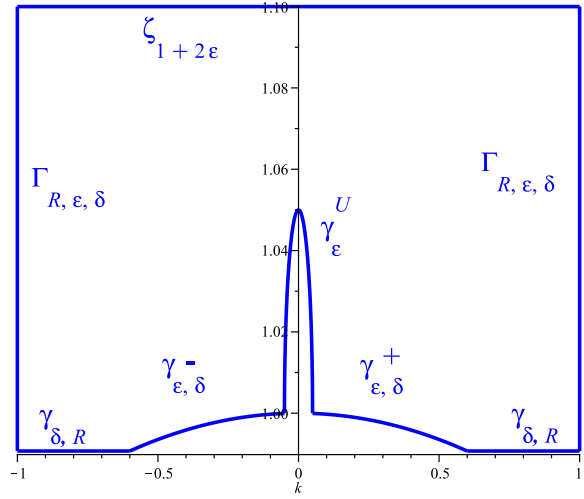}
\caption{Here we have plotted the closed contour of integration in
\eqref{Eq:Cauchy} for $\epsilon=0.05,\ \delta=0.6,\ R=1$.}
\label{fig:keyhole}
\end{figure}
By Cauchy's integral theorem, we now have
\begin{equation}\label{Eq:Cauchy}
\left(\int_{\gamma_{\delta,R}}+\int_{\gamma_{\epsilon,\delta}} \,+\int_{\gamma_\epsilon^U}+ \int_{\Gamma_{R,\epsilon,\delta}^{\pm}} -\int_{\zeta_{1+2\epsilon}}\right)\frac{\phi_t(-k)}{\I k-k^2}\exp\left(\I k\bar{\theta}t/2\right)\D k=0,
\end{equation}
Recall that the curves
$\zeta_{1+2\epsilon}$
and
$\gamma_{1+2\epsilon}$
are defined in~\eqref{ContourZeta} and~\eqref{ContourGamma}
respectively
and rewrite \eqref{Eq:Cauchy} as
\begin{align}\label{Eq:Cauchee}
\left(\int_{\gamma_{1+2\epsilon}}+\int_{\zeta_{1+2\epsilon}}\right)\frac{\phi_t(-k)}{\I k-k^2}\E^{\I k\bar{\theta}t/2}\D k & = \left(\int_{\gamma_{\delta,R}}+ \int_{\Gamma_{R,\epsilon,\delta}^{\pm}} +\int_{\gamma_{1+2\epsilon}}+\int_{\gamma_\epsilon^U}\right)\frac{\phi_t(-k)}{\I k-k^2}\E^{\I k\bar{\theta}t/2}\D k\nonumber \\
& + \int_{\gamma_{\epsilon,\delta}}\frac{\phi_t(-k)}{\I k-k^2}\E^{\I k\bar{\theta}t/2}\D k.
\end{align}
The integral on the left-hand side is equal to the normalised call price $2\pi S_0^{-1} \E^{-\bar{\theta}t/2}\mathbb{E}\left(S_t-S_0 \E^{\bar{\theta}t/2}\right)^{+}$ by Theorem 5.1 in Lee \cite{Lee04}, which is independent of $\epsilon$ (this holds because $1+2\epsilon<p_+$). For $k$ close to $\I$, we have
$$\frac{\phi_t(-k)}{\I k-k^2}\exp\left(\I k\bar{\theta}t/2\right) = \left(\frac{\I}{k-\I}+O\left(1\right)\right)\E^{-\bar{\theta}t/2},$$
so that
\begin{equation}\label{eq:Keyhole}
\int_{\gamma_\epsilon^U} \frac{\phi_t(-k)}{\I k-k^2}\exp\left(\I k\bar{\theta}t/2\right)\D k=\left(\pi+O\left(\epsilon\right)\right)\E^{-\bar{\theta}t/2}.
\end{equation}
Lemma \ref{LemmaSpecialCase} gives the behaviour of the last integral on the rhs of \eqref{Eq:Cauchee} as $\epsilon$ tends to $0$ for $\delta$ small enough. 

The other integrals can be bounded as follows.
By Lemma \ref{LemmaTailEstimates}, the integral along $\gamma_{1+2\epsilon}$ is $O\left(\E^{-\bar{\theta}t/2}\right)$, for $t>t^*(m)$, $R>R(m)$, as $\epsilon$ tends to $0$.

The curves
$\Gamma_{R,0,\delta}^{\pm}$ 
are both vertical strips of length 
$\delta$ 
and therefore their images are compact sets. 
Applying the tail estimate of 
Lemma~\ref{LemmaTailEstimates} 
along 
$\Gamma_{R,0,\delta}^{\pm}$, 
we know that for any $m>\bar{\theta}/2$, there exist $t(m)$ and $R(m)$ such that
$$\int_{\Gamma_{R,0,\delta}^{\pm}} \frac{\phi_t(-k)}{\I k-k^2}\E^{\I k\bar{\theta}t/2}\D k=O\left(\E^{-mt}\right),\quad\text{for all }t>t(m),|k|>R(m).$$

Lemma~\ref{lem:FAlongHorizontalContour}
implies that 
the real function 
$k_r\mapsto\Re\left(-\I (k_r+\I k_i)\bar{\theta}/2-V(-\I (k_r+\I k_i))\right)$ 
attains its 
global minimum 
at
$0$ 
for any fixed
$k_i\in(p_-,p_+)$
and is strictly decreasing (resp. increasing)
for 
$k_r<0$
(resp.
$k_r>0$).
It therefore follows that 
the function 
$u\mapsto\Re\left(-\I \gamma_{\delta,R}(u)\bar{\theta}/2-V(-\I \gamma_{\delta,R}(u))\right)$,
where
$u\in[-R,-\delta]\cup[\delta,R]$,
attains its minimum value
$g(\delta):=\Re\left((k_i(\delta)-\I\delta)\bar{\theta}/2-V(k_i(\delta)-\I \delta)\right)$,
where 
$k_i(\delta)$
is defined in~\eqref{eq:kiofs},
at the points
$u=\pm\delta$.
It can be checked directly that
$g(0)=\bar{\theta}/2$,
$g'(0)=0$
and
$g''(0)>0$
and hence 
for every
$\delta>0$
there exists
$\epsilon_0>0$ 
such that the following inequality holds
$$
\Re\left(-\I \gamma_{\delta,R}(u)\bar{\theta}/2-V(-\I \gamma_{\delta,R}(u))\right)>\bar{\theta}/2+\epsilon_0\quad
\text{for all}\quad u\in[-R,-\delta]\cup[\delta,R].
$$
Therefore Lemma~\ref{lem:AsymptoticBehaviourPhi} yields the following
inequality
$$\left|\int_{\gamma_{\delta,R}}\frac{\phi_t(-k)}{\I k-k^2}\E^{\I k\bar{\theta}t/2} \D k\right|\leq\E^{-(\bar{\theta}/2+\epsilon_0)t}\int_{\gamma_{\delta,R}}\left|\frac{U(-\I k)(1+\epsilon(k,t))}{\I k-k^2}\right|\D k=O\left(\exp\left(-\left(\bar{\theta}/2+\epsilon_0\right)t\right)\right).$$

We now prove the following lemma about the integral along $\gamma_{\epsilon,\delta}$ as $\epsilon$ tends to $0$.
\begin{lemma}\label{LemmaSpecialCase}
For 
$\delta>0$ 
and
sufficiently small we have
$$\lim_{\epsilon\searrow0}\int_{\gamma_{\epsilon,\delta}}\frac{\phi_t(-k)}{\I k-k^2}\exp\left(\I k\bar{\theta}t/2\right)\D k=
\sqrt{\frac{2\pi}{V''(1)t}}\E^{-\bar{\theta}t/2}\left(-1-\frac{1}{6}\frac{V'''(1)}{V''(1)}+U'(1)\right)\left(1+O(1/t)\right),$$
where $U$ is given by \eqref{eq:DefOfU} and $V$ by \eqref{eq:V(p)}.
\end{lemma}

\begin{proof}
Recall that 
$\gamma$
is the contour of steepest descent defined in
Lemma~\ref{lem:LemmaContour}
and that the curve
$\gamma_{\epsilon,\delta}$
is its restriction to the intervals
$[-\delta,-\epsilon]\cup[\epsilon,\delta]$.
Note that 
$s\mapsto\Re\left(\frac{\phi_t(-\gamma(s))\gamma'(s)}{\I\gamma(s)-\gamma(s)^2}\E^{\I\gamma(s)\bar{\theta}t/2}\right)$
is an even function and 
$s\mapsto\Im\left(\frac{\phi_t(-\gamma(s))\gamma'(s)}{\I\gamma(s)-\gamma(s)^2}\E^{\I\gamma(s)\bar{\theta}t/2}\right)$
is an odd function. We therefore obtain
\begin{equation}\label{EQSpecialCase1}
\int_{\gamma_{\epsilon,\delta}}\frac{\phi_t(-k)}{\I k-k^2}\exp\left(\I k\bar{\theta}t/2\right)\D k=\int_{[-\delta,-\epsilon] \cup [\epsilon,\delta]} \Re\left(\frac{\phi_t(-\gamma(s))\gamma'(s)}{\I\gamma(s)-\gamma(s)^2}\exp\left(\I\gamma(s)\bar{\theta}t/2\right)\right)\D s.
\end{equation}
From \eqref{eq:kiofs}, for $s$ around $0$, we have $\left(\I\gamma(s)-\gamma(s)^2\right)^{-1}=\I/s-1+O(s)$, so $\Re\left(\left(\I\gamma(s)-\gamma(s)^2\right)^{-1}\right)=-1+O(s)$,
i.e. taking the real part removes the singularity at $k=\I$.
Using Lemma \ref{lem:AsymptoticBehaviourPhi}, we then have
$$\int_{-\delta}^{\delta}\Re\left(\frac{\phi_t(-\gamma(s))\gamma'(s)}{\I\gamma(s)-\gamma(s)^2}\E^{\I\gamma(s)\bar{\theta}t/2}\right)\D s=\int_{-\delta}^{\delta}\Re\left(q(s)\right)\E^{\I\gamma(s)\bar{\theta}t/2+V\left(-\I\gamma(s)\right)t}\D s+O\left(\E^{-mt}\right),$$
for some $m>0$ large enough, where we define the function $q:\mathbb{R}\backslash\{0\}\to\mathbb{C}$ by
$$q(s):=\frac{U(-\I\gamma(s))\gamma'(s)}{\I\gamma(s)-\gamma(s)^2},\quad\text{for all }s\in\mathbb{R}.$$
Then, from \eqref{eq:kiofs}, we have the following expansion
\begin{equation}\label{TaylorSeriesQ}
q(s)=\left(\frac{\I}{s}-\left(\frac{V'''(1)}{6V''(1)}+1\right)\right)\left(1-\I U'(1)s\right)+O\left(s^3\right).
\end{equation}
We can therefore extend the function
$q$
to the map
$q:B_{\delta}(0)\backslash\{0\}\to\mathbb{C}$
for some 
$\delta>0$,
where
$B_{\delta}(0):=\{z\in\mathbb{C}:|z|<\delta\}$ 
is an open disc of radius 
$\delta$.
Note that for 
$s\in\mathbb{R}$
we have
$\Re\left(q(s)\right)=-1-V'''(1)/\left(6V''(1)\right)+U'(1)+O(s)$ 
and hence the function
$\Re(q):[-\delta,\delta]\to\mathbb{R}$
does not have a
singularity at $s=0$. 

Recall that if a function $G:B_{\delta}(0)\backslash\{0\}\to\mathbb{C}$ 
has a Laurent series expansion
$$G(z)=\frac{\I a_{-1}}{z}+ \sum_{n=0}^{\infty} a_n z^n,\quad\text{for }z\in B_{\delta}(0)\setminus\{0\},$$
with $a_{-1}\in\mathbb{R}$, 
then the function
$\Re(G):\mathbb{R} \cap B_{\delta}(0)\to\mathbb{R}$ 
has an analytic continuation on the whole disc
$B_{\delta}(0)$.
It follows from~\eqref{TaylorSeriesQ} 
that 
there exists a holomorphic function
$Q:B_\delta(0)\to\mathbb{C}$ 
such that 
$Q(s)=\Re(q(s))$ for any $s\in(-\delta,\delta)$.
Thus by Theorem 7.1, Chapter 4 of \cite{Olv74}, we have
\begin{align*}
\int_{-\delta}^{\delta} \Re\left(q(s)\right)\E^{\I\gamma(s)\bar{\theta}t/2+V\left(-\I\gamma(s)\right)t} \D s & = \int_{-\delta}^{\delta} Q(s)\E^{\I\gamma(s)\bar{\theta}t/2+V\left(-\I\gamma(s)\right)t}\D s\\
 & = \sqrt{\frac{2\pi\,\E^{-\bar{\theta}t}}{V''(1)t}}\left(-1-\frac{1}{6}\frac{V'''(1)}{V''(1)}+U'(1)\right)\left(1+O(1/t)\right).
\end{align*}
\end{proof}

Letting $\epsilon$ go to $0$ in equation \eqref{Eq:Cauchee}, 
applying Lemma~\ref{LemmaSpecialCase}
and the bounds developed above for the other integrals in \eqref{Eq:Cauchee} 
the theorem follows in the case 
$x=\bar{\theta}/2$. 
The case $x=-\theta/2$ is analogous.

\appendix
\renewcommand{\theequation}{A-\arabic{equation}}
\setcounter{equation}{0}  
\section*{APPENDIX}  

\section{Proof of Lemma \ref{lem:LargeTFixedK}}
\label{section:ProofofFixedStrike}
The proof is analogous to the proof of Theorem \ref{thm:HestonLargeT}. The residue in Theorem \ref{thm:HestonLargeT} is equal to $1$ (arising from the $\ind_{\{-\theta/2<x<\bar{\theta}/2\}}$ term), and from \eqref{eq:I}, the integral part is equal, for $R$ large enough, reads
$$\frac{\exp(xt)}{2\pi}\Re\left(\left(\int_{\zeta_{p^*(0)}}
+\int_{\gamma_{p^*(0)}}\right)\frac{\phi_t(-k)}{\I k-k^2} \E^{\I k xt}\D k\right),$$
and the behaviour of these integrals follows exactly the lines of the proof of Theorem \ref{thm:HestonLargeT}.

\section{Proof of Proposition \ref{prop:BStimedep}}\label{ProofLemma24}
Let us now consider a squared volatility of the form $\hat{\sigma}^2_t=\sigma^2+a_1/t>0$, then the Black-Scholes call option reads
\be\label{eq:BSFormula}
\frac{1}{S_0}C_{\mathrm{BS}}\left(S_0,S_0\E^{xt},t,\hat{\sigma}_t\right)=\Phi\left(\frac{-x+\left(\sigma^2+a_1/t\right)/2}{\sqrt{\sigma^2+a_1/t} }\sqrt{t}\right)-\E^{xt}\Phi\left(\frac{-x-\left(\sigma^2+a_1/t\right)/2}{\sqrt{\sigma^2+a_1/t}}\sqrt{t}\right),
\ee
and let $z_\pm=\left(-x\pm\frac{1}{2}\left(\sigma^2+a_1/t\right)\right)\sqrt{t}/\sqrt{\sigma^2+a_1/t}$.
Recall that (see \cite{Olv74})
\begin{equation}\label{OlverApproxPhi}
\Phi(-z)=1-\Phi(z)=\frac{\exp\left(-z^2/2\right)}{z\sqrt{2\pi}}\left(1+O\left(1/z^2\right)\right),\quad\text{ as }z\to+\infty
\end{equation}
\textbf{The case $x>\sigma^2/2$}. As $\sigma^2/2=\lim\limits_{t\to\infty}\hat{\sigma}^2_t/2$,
there exists $t^*$ such that for all $t>t^*$, $x>\hat{\sigma}^2_t/2=(\sigma+a_1/t)^2/2$.
From \eqref{eq:BSFormula}, we have, using a Taylor expansion for $z_\pm$,
\begin{eqnarray*}
\frac{1}{S_0}C_{\mathrm{BS}}(S_0,S_0 \E^{xt},t,\hat{\sigma}_t)
&=&\E^{\frac{1}{8}a_1\left(4x^2/\hat{\sigma}^4-1\right)}\left(\frac{\sigma}{x-\sigma^2/2}-\frac{\sigma}{x+\sigma^2/2 }\right)\frac{1}{\sqrt{2\pi t}}\exp\left(-\frac{(-x+\sigma^2/2)^2 t}{2\sigma^2}\right)\left(1+O(1/t)\right)\\
&=&\left(2\pi t\right)^{-1/2}\exp\left(-(V_{\mathrm{BS}}^*(x,\sigma)-x)t\right)A_{\mathrm{BS}}(x,\sigma,a_1)\left(1+O(1/t)\right).
\end{eqnarray*}
The cases $x<-\sigma^2/2$ and $-\sigma^2/2<x<\sigma^2/2$ follow likewise.\\
\textbf{The case $x=\sigma^2/2$}. From \eqref{eq:BSFormula}, we have
\begin{eqnarray*}
\frac{1}{S_0}C_{\mathrm{BS}}\left(S_0,S_0\E^{\sigma^2t/2},t,\hat{\sigma}_t\right)&=&\Phi\left(\frac{a_1/2}{\sqrt{\sigma^2t+a_1}}\right)-\frac{\E^{\sigma^2t/2}\sqrt{\sigma^2t+a_1}}{xt+(\sigma^2t+a_1)/2}\E^{-\frac{(-xt-(\sigma^2t+a_1)/2)^2}{2(\sigma^2t+a_1)}}(1+O(1/t))\\
&=& \frac{1}{2}+ \frac{a_1/2}{\sigma\sqrt{2\pi t}}-\frac{1}{\sigma\sqrt{2\pi t}}(1+O(1/t))= \frac{1}{2}+ \frac{A_{\mathrm{BS}}\left(\sigma^2/2,\sigma,a_1\right)}{\sqrt{2\pi t}}(1+O(1/t)).
\end{eqnarray*}
The case $x=-\sigma^2/2$ is analogous.
\end{document}